\documentclass[conference]{IEEEtran}

\usepackage{comment}
\usepackage[utf8]{inputenc}
\usepackage{amsmath,amssymb,bm}
\usepackage[dvips]{graphicx}
\usepackage{array}
\usepackage{cite}
\usepackage{enumerate}
\usepackage{booktabs}
\usepackage[draft]{hyperref} 
\usepackage{color}
\usepackage{float}
\usepackage[capitalise]{cleveref}
\usepackage[mathcal]{euscript}
\usepackage{multirow,multicol,tabu}
\usepackage{arydshln} 
\usepackage{subcaption}
\usepackage[left=0.625in,right=0.625in,top=0.75in,bottom=1in]{geometry}
\usepackage{soul}
\usepackage{url}
\usepackage{algorithm2e}
\SetKwComment{Comment}{/* }{ */}
\RestyleAlgo{ruled}
\newcolumntype{?}{!{\vrule width 1pt}}
\usepackage{pgfplots}
\usepackage{filecontents}

\usepackage{mathtools}
\usepackage{tikz,tabularx}

\usepackage{etoolbox}

\makeatletter
\newcommand{\changeoperator}[1]{%
  \csletcs{#1@saved}{#1@}%
  \csdef{#1@}{\changed@operator{#1}}%
}
\newcommand{\changed@operator}[1]{%
  \mathop{%
    \mathchoice{\textstyle\csuse{#1@saved}}
               {\csuse{#1@saved}}
               {\csuse{#1@saved}}
               {\csuse{#1@saved}}%
  }
}
\makeatother

\changeoperator{sum}
\changeoperator{prod}

\usetikzlibrary{decorations.pathreplacing}
\tikzstyle{every picture}+=[remember picture]

\usetikzlibrary{shapes,arrows,shadows}
\usepgflibrary{patterns}
\usetikzlibrary{patterns}
\usetikzlibrary{shapes.geometric}
\usetikzlibrary{arrows}

\usepackage{tikz,tabularx}
\usetikzlibrary{decorations.pathreplacing}
\tikzstyle{every picture}+=[remember picture]

\usetikzlibrary{decorations.markings}

\usepackage{multicol}
\usepackage{pgfplots}
\pgfplotsset{major grid style={densely dotted,black!70}}
\pgfplotsset{compat=1.9}
\usetikzlibrary{arrows}

\usepackage{amsthm}

\newtheorem{remark}{Remark}

\newtheorem{theorem}{Theorem}
\newtheorem{lemma}{Lemma}
\newtheorem{proposition}{Proposition}

\newcommand{\tcb}[1]{\textcolor{blue}{#1}}

\newcommand \f{\footnotesize}

\newcommand \m{\small}

\usepackage{tikz,tabularx}
\usetikzlibrary{decorations.pathreplacing}
\usetikzlibrary{shapes.geometric, arrows.meta, fit, positioning}
\tikzstyle{every picture}+=[remember picture]

\usetikzlibrary{shapes,arrows,shadows}
\usepgflibrary{patterns}
\usetikzlibrary{patterns}
\usetikzlibrary{shapes.geometric}
\usetikzlibrary{arrows}
\usetikzlibrary{positioning}

\usepackage{tikz,tabularx}
\usetikzlibrary{decorations.pathreplacing}
\tikzstyle{every picture}+=[remember picture]

\usetikzlibrary{decorations.markings}

\usepackage{multicol}
\usepackage{pgfplots}
\usepackage{mathtools}  
\newcommand{\qllr}[1][k]{\Re\left\{(\tilde r_{#1})^{\,M}\right\}}
\definecolor{deepbrown}{RGB}{184, 130, 13}

\title{Sequential Change Point Detection for PSK Modulation Switching}

\author{
\IEEEauthorblockN{Gopal Kataria}
\IEEEauthorblockA{\textit{IIIT Hyderabad} \\
gopal.kataria@research.iiit.ac.in}
\and
\IEEEauthorblockN{Davin Sequeira}
\IEEEauthorblockA{\textit{IIIT Hyderabad} \\
davin.sequeira@research.iiit.ac.in}
\and
\IEEEauthorblockN{Arti Yardi}
\IEEEauthorblockA{\textit{IIIT Hyderabad} \\
arti.yardi@iiit.ac.in}
}

\date{June 2026}

\usepackage{xcolor}

\definecolor{draftgray}{gray}{0.55}

\begin{document}

\maketitle


%

\begin{abstract}
We study the problem of sequential change-point detection (CPD) in the received data when the modulation scheme of the transmitter switches from one phase shift keying (PSK) to another over an additive white Gaussian noise channel. Such scenarios arise in wireless communication systems when a transmitter deploys adaptive modulation and coding (AMC) and  dynamically switches between different PSKs. Towards this, we  use two classical CPD algorithms; cumulative-sum (CUSUM) and moving-sum algorithm (MOSUM), implemented using the log-likelihood ratio (LLR) and its approximation.
For a switch from quadrature to binary PSK, we derive the moment-generating function and probability density function of the approximate LLR and use them to relate the CUSUM threshold to the average run length to false alarm (ARL2FA) and to the 
worst-case average detection delay
(WADD).
For a switch between any two PSK orders, we approximate the vector of overlapping MOSUM statistics as multivariate Gaussian and obtain bounds on ARL2FA and WADD. Monte Carlo simulations show that our derived expressions for ARL2FA and WADD closely match the empirical results across thresholds, validating the accuracy of the proposed analytical approximations and theoretical bounds.
\end{abstract}


\section{Introduction}
\label{Sec:Introduction}

In wireless communication systems, Adaptive modulation allows the transmitter to vary its modulation scheme according to the underlying channel conditions, thereby improving spectral efficiency and link reliability \cite{Goldsmith_Book}. Automatic modulation classification (AMC) enables the receiver to identify the modulation scheme directly from the received signal and has applications in cognitive radio, spectrum monitoring, and non-cooperative communication systems. AMC has been studied extensively with many works using different methods such as ALRT  \cite{SapianoMartin1996,PSK_classification_1995}, GLRT \cite{PolydorosKim1990}, HLRT \cite{AbdiDobre2004, Panagiotou2000,PolydorosKim1990}, Cumulants \cite{SwamiSadler2000,DobreBarNessSu2003} and other methods \cite{ZhuNandi2015}. 
Previously, change-point detection has been studied in channel-coded communication, particularly for detecting transitions between binary linear block codes \cite{vats2026neural}. A related problem arises when the modulation scheme changes during transmission, as the receiver must identify the time instant at which the transition occurs. This motivates the study of online change-point detection (CPD) for modulation switching, where the objective is to detect such transitions from sequentially observed signals while controlling the false-alarm rate and detection delay. While prior work has considered sequential detection for automatic modulation classification (AMC) using LSTMs \cite{waltman2019sequential}, to the best of our knowledge, this is the first work to investigate online CPD for identifying transitions between modulation schemes.

In this work, we consider a communication system in which the transmitter switches from $M^{\prime}$-PSK to $M$-PSK at an unknown time instant. The transmitted signal passes through an additive white Gaussian noise (AWGN) channel, and the receiver processes the observations sequentially to detect the modulation change. We formulate this problem as a sequential CPD problem and employ cumulative-sum (CUSUM) and moving-sum (MOSUM) detectors based on an exact and approximate log-likelihood ratio (LLR) statistic. These statistics are motivated by the likelihood-based PSK classification method proposed in \cite{PSK_classification_1995}, which develops a low SNR approximation to the likelihood functional for MPSK signals.

For the case $(M^\prime=4)$ and (M=2), closed-form expressions for the probability density function (PDF) and moment-generating function (MGF) of the approximate LLR are derived analytically. These expressions are used to characterize the false-alarm and detection-delay performance of the CUSUM detector. For higher-order PSK transitions, however, deriving closed-form expressions for the PDF or MGF of the approximate LLR becomes challenging. We therefore consider a MOSUM detector and use the first two moments of the approximate LLR to construct a multivariate Gaussian approximation to the vector of moving sums.

The main contributions of this work are summarized as follows:
\begin{itemize}
    \item We formulate the detection of a transition from
    $M'$-PSK to $M$-PSK as a sequential CPD problem and
    construct CUSUM and MOSUM detectors using exact and
    approximate LLR statistics.

    \item For the QPSK-to-BPSK transition, we derive the PDF
    and MGF of the approximate LLR and use them to analyze
    the false-alarm and detection-delay performance of CUSUM.

    \item For higher-order PSK transitions, we characterize
    the first two moments of the approximate LLR and develop
    a multivariate Gaussian approximation for the MOSUM
    statistic vector that accounts for dependence between
    overlapping windows.
\end{itemize}

\textbf{Organization:} The remainder of this paper is organized as follows. In \cref{Sec:System_model}, we describe the system model and formulate the change-point detection problem. In \cref{Sec:Proposed_method}, we present the proposed detection methods. Their performance is analyzed in \cref{Sec:Performance}. Numerical results are presented in \cref{Sec:Simulations}, followed by concluding remarks in \cref{Sec:Conclusion}.


\section{System Model and Problem Formulation}
\label{Sec:System_model}

%
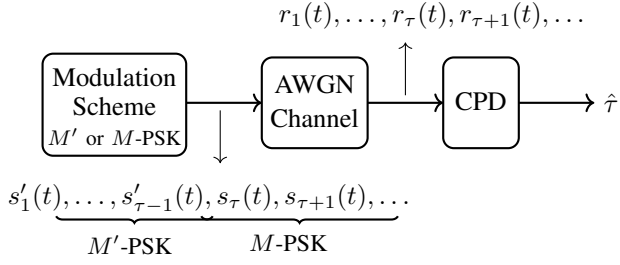
\begin{figure}[tp]
\begin{center}
    \begin{tikzpicture}
    
    \draw [rounded corners,thick] (0.3,-0.7) rectangle (0.5+1.6+0.1,0.65);
    \node [above] at (0.35+0.8+0.1,0.1) {Modulation};
    \node [above] at (0.35+0.8+0.1,0.1-0.39) {Scheme};
    \node [above] at (0.35+0.8++0.1,0.1-0.39-0.39) {\f $M^{\prime}$ or $M$-PSK};
    
    \draw [->] (2.65,-0.1) -- (2.65,-0.1-0.7);
    
    \node [below] at (2.5,-1.3+0.35) {$s_1^{\prime}(t), \ldots,  s_{\tau-1}^{\prime}(t), s_{\tau}(t),s_{\tau+1}(t), \ldots $};
    
    \draw[decorate,decoration=brace, thick] (6.6-0.2-4.65+0.5+0.3,-2+0.5)--(4.65-0.05-4.5+0.5-0.12,-2+0.5);
    
    \draw[decorate,decoration=brace, thick] (8.6-0.2-4.5+1.1,-2+0.5)--(6.85-0.25-4.6+0.4,-2+0.5);
    
    \node [below] at (5.65-0.2-4.5+0.5,-2.15+0.53) {\m $M^{\prime}$-PSK};
    
    \node [below] at (7.75-0.2-4.5+0.5,-2.15+0.53) {\m $M$-PSK};
    
    \draw [->,thick] (0.5+1.7,0) -- (0.5+1.7+1,0);
    
    \draw [rounded corners,thick] (0.5+1.7+1,-0.6) rectangle (0.5+1.7+1+1.4,0.6);
    
    \node [above] at (0.5+1.7+1+0.7,0) {AWGN};
    \node [above] at (0.5+1.7+1+0.7,-0.45) {Channel};

    \draw [->] (0.5+1.7+1+1.4+0.5,0.1) -- (0.5+1.7+1+1.4+0.5,0.1+0.7);

    \node [above] at (0.5+1.7+1+1.4+0.5+0.35,0.1+0.7+0.06) {$r_1(t),  \ldots,  r_{\tau}(t), r_{\tau+1}(t), \ldots$};
        

    \draw [->,thick] (0.5+1.7+1+1.4,0) -- (0.5+1.7+1+1.4+1,0);

    \draw [rounded corners,thick] (0.5+1.7+1+1.4+1,-0.6) rectangle (0.5+1.7+1+1.4+1+1,0.6);

    \node [above] at (0.5+1.7+1+1.4+1+0.5,-0.22) {CPD};    
    \draw [->,thick] (0.5+1.7+1+1.4+1+1,0) -- (0.5+1.7+1+1.4+1+1+1,0);

    \node [right] at (0.5+1.7+1+1.4+1+1+1, 0) {$\hat{\tau}$};

    \end{tikzpicture}

\end{center}

\caption{System model for the sequential change-point detection (CPD) problem for modulation schemes}

\label{Figure_system_model_M1M2}

\end{figure}

In this work, we consider a setup where the transmitter uses the phase shift keying (PSK) modulator and at some unknown time instance $\tau$ it switches from $M^{\prime}$-PSK to $M$-PSK  (see \cref{Figure_system_model_M1M2}). 
Let $s_k^{\prime}(t)$ and $s_k(t)$ denote the modulated waveforms corresponding to the $k$-th symbol index for $M^{\prime}$-PSK and $M$-PSK respectively.
Let $T_s$ be the symbol duration. Note that for the $k$-th symbol, $s_k(t)$ corresponds to the transmitted signal waveform for the time duration $t=(k-1)T_s$ to $t=kT_s$. 
As illustrated in the figure, $s_1^{\prime}(t), s_2^{\prime}(t), \ldots,  s_{\tau-1}^{\prime}(t)$ denote the waveforms corresponding to $M^{\prime}$-PSK and $s_{\tau}(t), s_{\tau+1}(t), \ldots $ denote the waveforms corresponding to $M$-PSK. We consider the passband signal based model for the PSK modulation and assume that 
the noise is introduced by the
additive white Gaussian noise (AWGN) channel
to get the received waveforms $r_1(t),  \ldots,  r_{\tau}(t), r_{\tau+1}(t), \ldots$.
Given this received data, our goal is to estimate the location $\tau$, where the switch between modulation schemes happen. We formulate and solve this problem as a sequential change-point detection (CPD) problem~\cite{CDP_book_Basseville_1993}.

For the $k$-th symbol index, let the transmitted waveform $s_k(t)$ corresponding to $M$-PSK is given by~\cite{Communications_book_proakis},
\begin{equation}
    s_k(t)
    =
    \operatorname{Re}
    \left\{
        \sqrt{2S}\,
        e^{j\theta_k}
        p(t-(k-1)T_s)
        e^{j(2\pi f_c t+\theta_c)}
    \right\},
\end{equation}
where $S$ denotes the signal power, $T_s$ is the symbol duration, $p(t)$ is the transmit pulse, $\theta_c$ is the carrier phase, and $\theta_k$ denotes the phase of the $k$-th symbol.
In this work we assume that the carrier phase $\theta_c$ is known to the receiver and hence, without loss of generality, it can be assumed to be equal to zero~\cite{PSK_classification_1995}. For $M$-PSK, the information-bearing phase $\theta_k$ is given by
\begin{equation}
    \theta_k \in
    \left\{
        \frac{2\pi m}{M}
        : m=0,\ldots,M-1
    \right\}.
\label{Eq:All_phases_MPSK}
\end{equation}
The expressions for the waveform of $M^{\prime}$-PSK can be written in a similar manner. 
%
%

The received signal $r_k(t)$ corresponding to the $k$-th symbol is given by
\begin{equation}
r_k(t) = 
\begin{cases} 
s_k^{\prime}(t) + n_k(t),     & \text{if } k < \tau \\
s_k(t) + n_k(t),     & \text{if } k \geq \tau ,
\end{cases}
\end{equation}
where $n_k(t)$ denote the AWGN with a two-sided power spectral density of $N_0/2$ W/Hz. Let $\gamma_s$ denotes the symbol signal-to-noise ratio (SNR), given by
\begin{equation}
\gamma_s = \frac{S T_s}{N_0}.
\label{eq:snr_def}
\end{equation}
Let $\tilde{r}_k$ denote the normalized complex-valued observation
obtained after down-conversion, pulse-matched filtering, and symbol-rate
sampling~\cite{Communications_book_proakis,PSK_classification_1995} given by,
\begin{align}
    \tilde{r}_k = \frac{1}{\sqrt{N_0/T_s}}\int_{(k-1)T_s}^{kT_s} r_k(t)\,p(t-(k-1)T_s)\,e^{-j2\pi f_c}\,dt.
    \label{Eq:r_k_tilda}
\end{align}
Note that $\tilde{r}_k$ is a complex number given by
\begin{align}
 \tilde{r}_k = \tilde{r}_{I,k} + j\tilde{r}_{Q,k},   
  \label{Eq:r_k_tilda_real_imag}
\end{align}
where $\tilde{r}_{I,k}$ and $\tilde{r}_{Q,k}$ denote the real and imaginary parts of $\tilde{r}_{k}$ respectively. 
Let $n_k$ denote the complex envelop of the noise sample corresponding to $n_k(t)$. Since $n_k$ is obtained via similar matched-filter normalization (see \cref{Eq:r_k_tilda}), the real and imaginary parts of $n_k$ are known to be independent and identically distributed (i.i.d.) according to $\mathcal{N}(0,1)$ distribution~\cite{Communications_book_proakis}. Further for the AWGN channel,  noise samples $n_k$ are independent across $k$.

For our CPD problem, we find the likelihood of receiving $r_k(t)$ under the condition a particular modularion scheme is used. 
Let $l^{k}_{M\text{-PSK}}$ and $l^{k}_{M^{\prime}\text{-PSK}}$ denote the log likelihood of receiving $r_k(t)$ under the condition $M$-PSK and $M^{\prime}$-PSK are used respectively. 
The likelihood $l^{k}_{M\text{-PSK}}$ can be computed by conditioning over all possible values $\theta_k$ can take (see Eq.~(3) of~\cite{PSK_classification_1995}) and is given by
\begin{equation}
\begin{aligned}
l^{k}_{M\text{-PSK}}
&\coloneqq \log \mathbb{P}\big[r_k(t)\mid M\big] \\
&= \ln \left[
\frac{2}{M}\sum_{i=0}^{M/2-1}
\cosh\!\left(
\sqrt{2\gamma_s}\,
\Re\!\left\{\tilde{r}_k e^{j2\pi i/M}\right\}
\right)
\right],
\end{aligned}
\label{eq:llf_mpsk}
\end{equation}
where $\Re\{\cdot\}$ denotes the real part of a complex number and the detailed steps towards the computation of
\cref{eq:llf_mpsk} can be found in \cite[Sec.~II]{PSK_classification_1995} (see Eq.~(3) and (4) of \cite{PSK_classification_1995}).
The likelihood $l^{k}_{M^{\prime}\text{-PSK}}$ follows from \eqref{eq:llf_mpsk} with $M$ replaced by $M^{\prime}$. Let $\ell_k$ denote the
log-likelihood ratio (LLR) of the observation $r_k(t)$ given by,
\begin{align}
    \ell_k \coloneqq \log\frac{\mathbb{P}[r_k(t)\mid M]}{\mathbb{P}[r_k(t)\mid M^{\prime}]}
    = l^{k}_{M\text{-PSK}} - l^{k}_{M^{\prime}\text{-PSK}},
    \label{eq:llr}
\end{align}
which, written out explicitly, is
\begin{multline}
    \ell_k
    = \ln\!\Bigg[\frac{2}{M}\sum_{i=0}^{M/2-1}
        \cosh\!\big(\sqrt{2\gamma_s}\,\Re\{\tilde{r}_k e^{j2\pi i/M}\}\big)\Bigg] \\
    - \ln\!\Bigg[\frac{2}{M^{\prime}}\sum_{i=0}^{M^{\prime}/2-1}
        \cosh\!\big(\sqrt{2\gamma_s}\,\Re\{\tilde{r}_k e^{j2\pi i/M^{\prime}}\}\big)\Bigg].
    \label{eq:llr_explicit}
\end{multline}

In \cite{PSK_classification_1995}, authors use the LLRs defined in \cref{eq:llr} to study the problem of classification between modulation schemes $M^{\prime}$-PSK and $M$-PSK. They identify that the expression in \cref{eq:llr} is not amenable to analysis and hence propose to use approximations for these LLRs, which are applicable only when $M^\prime \geq 2M$. They approximated LLR $\ell_k$ defined in \cref{eq:llr} by $\lambda_k$ given by
\begin{align}
  \ell_k \approx  \lambda_k = \qllr 
\label{eq:approx_llr}
\end{align}
where recall that $\tilde{r}_k$ is defined in  \cref{Eq:r_k_tilda} and the detailed proof of the last step is provided in \cite[Sec.~III-A]{PSK_classification_1995} (see Eq.~(19) of \cite{PSK_classification_1995}).

For our problem of change-point detection corresponding to the received data $r_1(t), r_2(t), \ldots$, we make use of the sequence of LLRs $\ell_1, \ell_2, \ldots$ and corresponding approximate LLRs
$\lambda_1, \lambda_2, \ldots$ and provide a method to estimate the location of the change-point $\tau$ in \cref{Sec:Proposed_method}. Analysis of our approach is discussed in Sections~\ref{Sec:QPSK_BPSK} and \ref{Sec:PSK}.

\section{Proposed methods for CPD}
\label{Sec:Proposed_method}

In this section, we discuss our proposed approach for finding the location of the change-point $\tau$
corresponding to the received data 
$r_1(t), r_2(t), \ldots$ (see \cref{Figure_system_model_M1M2}). 
Since the received data is observed in a sequential manner, our problem  corresponds to the problem of \textit{sequential} CPD. While in this section we discuss the proposed method, in \cref{Sec:Performance} we study its analytical performance guarantees.

For the received data $r_1(t), r_2(t), \ldots$, 
we first compute the corresponding sequence of LLRs, $\ell_1, \ell_2, \ldots$, using \cref{eq:llr}, as well as the sequence of approximate LLRs, $\lambda_1, \lambda_2, \ldots$, using \cref{eq:approx_llr}. These sequences are then used to estimate the change-point location $\tau$.
We would like to mention here that the LLRs $\ell_1, \ell_2, \ldots$ form a sufficient statistics for our setup~\cite{CDP_book_Basseville_1993} and hence the performance of any classical CPD method constructed using $\ell_k$ will yield the optimal performance. However, as noted by the authors of \cite{PSK_classification_1995}, the distribution of $\ell_k$ is generally difficult to characterize, which makes the analysis of a CPD method based on $\ell_k$ challenging.
Hence we construct CPD methods using both $\ell_k$ and $\lambda_k$. 

In our work, we use the classical cumulative sum (CUSUM) procedure and moving sum (MOSUM) procedure to estimate the change-point~\cite{CDP_book_Basseville_1993}, which are described next. 
%

\subsection{CUSUM method for CPD}
\label{subsection_CUSUM_algo}

The CUSUM method is a classical method for sequential CPD problem~\cite{page_cusum, CDP_book_Basseville_1993}. 
We next describe the key idea in the CUSUM method for CPD.
The CUSUM method sequentially accumulates the LLRs $\ell_k$, as the data arrives. It can be observed from \cref{eq:llr} that the cumulative LLR should have a negative drift before the change-point $\tau$ and a positive drift after $\tau$. Consequently, if the accumulated LLR is negative at the time of the change, it may take some time to return to zero and begin accumulating positive values, resulting in a detection delay. To mitigate this effect, the CUSUM method resets the statistic to zero whenever it becomes negative, thereby allowing the accumulation of post-change LLR values to begin afresh.

Let $W_k^{\ell}$ denotes the accumulated LLRs till time $k$, for the CUSUM method. We will also refer to $W_k$ as the CUSUM statistic. The CUSUM method begins with $W_0^{\ell} = 0$ and for each $\ell_k$ received sequentially, $W_k^{\ell}$ is computed using $W_{k-1}^{\ell}$ as follows
\begin{align}
    W_k^{\ell}=\max\{0,\,W_{k-1}^{\ell}+\ell_k\}.    
    \label{Eq:Wk_recursion}
\end{align}
For a suitable positive threshold $h_c^{\ell}$, the estimated change-point (or the stopping time) $T_c^{\ell}$ for the CUSUM method is given by 
\begin{align}    
    T_c^{\ell}=\inf\{k\geq1:\ W_k^{\ell}\geq h_c^{\ell}\}.
    \label{Eq:Stopping_CUSUM_exact}
\end{align}
Note that, in the notation $h_c^{\ell}$ (or $T_c^{\ell}$), the subscript $c$ denotes the CUSUM method, while the superscript $\ell$ indicates that the LLRs $\ell_k$ are used for CP estimation.

A similar CUSUM procedure can be constructed using the approximate LLRs $\lambda_k$. For the CUSUM implemented using $\lambda_k$, let $W_k^{\lambda}$,  $h_c^{\lambda}$, and $T_c^{\lambda}$ denote the CUSUM statistic, threshold, and stopping time respectively. 
In \cite{PSK_classification_1995} it was shown that, the expected value of $\lambda_k$ before CP is equal to zero, i.e., $\mathbb{E}_{M^{\prime}}[\lambda_k]=0$. Since 
the standard CUSUM relies on a negative drift before the change point~\cite{CDP_book_Basseville_1993}, $\mathbb{E}_{M^{\prime}}[\lambda_k]=0$ causes a zero drift before CP. Due to this, $W_k^{\lambda}$ can cross a positive threshold even though no change has occurred, leading to a false detection.  
To tackle this issue, typically a negative bias is added to $W_k^{\lambda}$~\cite{CDP_book_Basseville_1993} to ensure the negative drift before CP. 
For our setup it is known that the expected value of $\lambda_k$ for post-change modulation scheme $M$-PSK is equal to $\mathbb{E}_{M}[\lambda_k]=(2\gamma_s)^{M/2}$~\cite[Eqs.~(29)--(32)]{PSK_classification_1995}. Hence we choose the mid-point $b=\tfrac{1}{2}(2\gamma_s)^{M/2}$ of the two expected values as a bias and the corresponding CUSUM procedure and stopping rule are given by  
\begin{equation}
\begin{aligned}
    W_k^{\lambda} &= \max\left\{0,\, W_{k-1}^{\lambda}+\lambda_k-b\right\},
    \\ 
    T_c^{\lambda}   &=  \inf\left\{k\geq1:\ W_k^{\lambda}\geq h_c^{\lambda}\right\}.   \label{Eq:Stopping_CUSUM_approx}
\end{aligned}
%
\end{equation}

\subsection{MOSUM method for CPD}
\label{subsection_MOSUM_algo}

The MOSUM procedure is another classical online change-point detection method~\cite{CDP_book_Basseville_1993}. 
While the CUSUM method accumulates all LLRs up to the current observation $k$, the MOSUM method considers a sliding window of size $w \geq 1$ and accumulates the most recent $w$ LLRs.
The MOSUM method is useful when the observations from the distant past have little influence on the stopping decision at time $k$.

At time $k$, the MOSUM method finds the statistic $Y_k^{\ell}$, termed as the MOSUM statistic, is defined as follows
\begin{equation}
    Y_k^{\ell} = \sum_{k^{\prime}=k-w+1}^{k}\ell_{k^{\prime}}, \label{Eq:MOSUM_statistic_exact}
\end{equation}
where $k=w, w+1, \ldots$ and the corresponding stopping time $h_m^{\ell}$ is given by
\begin{equation}
    T_m^{\ell} = \inf\left\{k\geq w:Y_k^{\ell}\ge h_m^{\ell}\right\}.\label{Eq:Stopping_MOSUM_exact}
\end{equation}
Note that in the notation $h_m^{\ell}$, similar to the notation of CUSUM method, the subscript $m$ is used to indicate MOSUM method. 
We  implement the MOSUM method  using approximate LLRs $\lambda_k$ as well and  the corresponding statistic, stopping time, and threshold  are denoted by
$Y_k^{\lambda}$, $T_m^{\lambda}$, and $h_m^{\lambda}$ respectively.

Note that in \cref{Eq:MOSUM_statistic_exact}, a weighted average of $\ell_d$ can be taken and this variant of CPD method is referred to as the finite moving average (FMA) chart method in the literature~\cite[Section 5.1.3]{CDP_book_Basseville_1993}. 
The MOSUM method 
is a special case of the FMA method where the weights corresponding to each LLR $\ell_i$ are set to $1$.




\subsection{Performance Measures}
\label{subsec:performance_measures}
The performance of a sequential CPD procedure is, in general, characterized by the two competing criteria: the frequency of false alarms before the change and the delay in detecting the change once it has occurred. We study the performance of our setup by analyzing the average run length to false alarm ($ARL2FA$), the average detection delay ($ADD$), and the worst-case average detection delay ($WADD$)~\cite[Ch.~4]{CDP_book_Basseville_1993}, which are defined in this subsection.
We define these quantities for the CUSUM CPD setup, implemented using the exact LLRs $\ell_k$. For other CPD methods, these quantities can be defined in an analogous manner. 


%
%
%

We first define the average run length (ARL). The $ARL$ is defined as the expected number of observations processed before the detector stops, denoted by $\mathbb{E}[T_c^{\ell}]$, where the expectation is taken under the hypothesis that governs the data. For our setup, the ARL before and after change-point will be $\mathbb{E}_{M^{\prime}}[T]$ and $\mathbb{E}_{M}[T]$ respectively. 

A false alarm occurs when a change is declared but the transmitter has not switched from $M^{\prime}$-PSK to $M$-PSK. The $ARL2FA$ is the average time the detector operates before raising such a false alarm, i.e., when the transmitter uses $M^{\prime}$-PSK for the entire time. A large $ARL2FA$ thus means that false alarms are rare. It is given by
\begin{equation}
    ARL2FA(T_c^{\ell}) = \mathbb{E}_{M^{\prime}}[T_c^{\ell}].
    \label{eq:ARL2FA_def}
\end{equation}

Once a change occurs at time $\tau$, we are interested in what is the expected delay between the change-point and the stopping time. For a fixed change-point $\tau$, this is the difference between the stopping time $T_c^{\ell}$ and $\tau$. The delay is meaningful only when we have not raised a false alarm, therefore we condition it on $T_c^{\ell}\geq\tau$. The average detection delay is
\begin{equation}
    ADD(\tau)
    =
    \mathbb{E}_{\tau}
    \left[
        T_c^{\ell}-\tau
        \,\middle|\,
        T_c^{\ell}\geq\tau
    \right],
\end{equation}
where $\mathbb{E}_{\tau}$ denotes the expectation when the transmitter switches from $M^{\prime}$-PSK to $M$-PSK at time $\tau$. 

Detection delay depends on the change point $\tau$ and the value of the statistic at the time of change, which is determined by pre-change observations. To account for the least favorable pre-change history possible, we consider the $WADD$ defined as
\begin{equation}
    WADD(T_c^{\ell})
    =
    \sup_{\tau\geq1}
    \operatorname*{ess\,sup}
    \mathbb{E}_{\tau}
    \left[
        (T_c^{\ell}-\tau)^+
        \,\middle|\,
        \mathcal{F}_{\tau-1}
    \right],
    \label{eq:WADD_def}
\end{equation}
where $(x)^+=\max(x,0)$ and $\mathcal{F}_{\tau-1}$ denotes the  $\sigma$-algebra generated by observations $\tilde{r}_1,\ldots,\tilde{r}_{\tau-1}$ received before the change. The WADD is thus the worst expected delay over all change times and pre-change histories.


\section{Performance analysis of the proposed CPD methods}
\label{Sec:Performance}

In this section, we study the performance of the CUSUM and MOSUM methods. As discussed earlier, the distribution of the LLRs $\ell_k$ is generally difficult to characterize~\cite{PSK_classification_1995} and hence for analysis, we focus on the CPD methods implemented using the approximate LLRs $\lambda_k$. 

For both CUSUM and MOSUM methods, the choice of thresholds $h_c^{\lambda}$ and $h_m^{\lambda}$ is a crucial design problem. While a higher threshold increases $ARL2FA$, it also leads to a larger detection delay, hence the two performance metrics cannot be simultaneously improved.
Further for practical applications, it is desirable to set a threshold that will fetch the user specified values of $ARL2FA$ or $WADD$. 
Hence, towards analyzing the performance of our methods, we aim to characterize the relation between any arbitrary choice of threshold and the corresponding values of $ARL2FA$ or $WADD$.   
In Sections~\ref{Sec:QPSK_BPSK} and \ref{Sec:PSK}, we study the performance of CUSUM and MOSUM methods respectively. 

For characterizing the performance of CUSUM, the distribution of $\lambda_k$ is required. Towards this, we derive the closed form expressions for the moment generating function (MGF) and the probability density function (PDF) of $\lambda_k$, for case when modulation scheme switches from QPSK to BPSK (see Theorems~\ref{Thm:MGF_QPSK_BPSK}, \ref{Thm:PDF_QPSK_BPSK}). We identify that deriving these for arbitrary PSKs become difficult (see \cref{Remark:Difficulty_in_any_MPSK}). For MOSUM, we derive the cumulative distribution function (CDF) of the vector corresponding to the MOSUM statistic (see \cref{thm:Gaussian_Approx}), which can be obtained for any PSK modulation scheme and use it to characterize $ARL2FA$ and $WADD$.

\subsection{Performance analysis of CUSUM for QPSK to BPSK}
\label{Sec:QPSK_BPSK}


For analyzing the performance of CUSUM, we use 
the analysis proposed by Wald and Siegmund~\cite[Sec.~5.2]{CDP_book_Basseville_1993}, which needs the MGF and PDF of $\lambda_k$
under both the hypotheses. 
We would like to mention here that, authors of \cite{PSK_classification_1995} have derived the expressions for the mean and variance of $\lambda_k$. 
In our work we derive the closed form expressions of MGF and PDF in Theorems~\ref{Thm:MGF_QPSK_BPSK} and \ref{Thm:PDF_QPSK_BPSK} respectively.


%


%
%
%


\begin{theorem}[MGF of $\lambda_k$]
\label{Thm:MGF_QPSK_BPSK}
Consider the setup when the modulation scheme of the transmitter switches from QPSK to BPSK. Then the MGF of the approximate LLRs $\lambda_k$ defined in \cref{eq:approx_llr} under the condition data is generated by QPSK is given by
\begin{align*}
    M_{\lambda}\big(t \mid \mathrm{QPSK}\big)
    ={}&
    \frac{1}{2\sqrt{1-4t^{2}}}
    \Bigg[
    \exp\!\left(\frac{2\gamma_s t}{1-2t}\right)
    \\
    &\quad
    +\exp\!\left(-\frac{2\gamma_s t}{1+2t}\right)
    \Bigg],
    \quad |t|<\tfrac12 .
\end{align*}
Further, when the data is generated by BPSK,
the MGF of $\lambda_k$ under BPSK is given by
\begin{align*}
    M_{\lambda}\big(t\big \vert BPSK\big)
    =
    \frac{1}{\sqrt{1-4t^2}}
    \exp\!\left(
    \frac{2\gamma_s t}{1-2t}
    \right),
    \qquad |t|<\tfrac{1}{2}.
\end{align*}
\end{theorem}
\begin{IEEEproof}
See \cref{Appendix:A_MGF_CUSUM}.
\end{IEEEproof}

\begin{theorem}[PDF of $\lambda_k$]
\label{Thm:PDF_QPSK_BPSK}
Consider the setup when the modulation scheme of the transmitter switches from QPSK to BPSK. Then, the PDF of the approximate LLR $\lambda_k$ defined in \cref{eq:approx_llr} when the data is generated by QPSK is given by
\begin{equation}
\begin{aligned}
    f_{\lambda_k}\big(x \mid \mathrm{QPSK}\big)
    ={}&
    e^{-\gamma_s}
    \sum_{n=0}^{\infty}
    \sum_{j=0}^{n}
    \frac{\binom{2n}{2j}}
    {\pi(2n)!\,2^{2n+1}}
    (2\gamma_s)^n
    \\
    &\times
    |x|^{n}\,
    K_{n-2j}\left(\frac{|x|}{2}\right).
    \label{eq:thm_qpsk_pdf}
\end{aligned}
\end{equation}
Further, when the data is generated by BPSK,
the PDF of $\lambda_k$ is given by
\begin{equation}
\begin{aligned}
    \label{eq:thm_bpsk_pdf}
    f_{\lambda_k}\big(x \mid \mathrm{BPSK}\big)
    ={}&
    e^{-\gamma_s}
    \sum_{n=0}^{\infty}
    \sum_{m=0}^{2n}
    \frac{\binom{2n}{m}}
    {\pi(2n)!\,2^{2n+1}}
    (2\gamma_s)^n
    \\
    &\times
    x^{2n-m}|x|^{m-n}
    K_{m-n}\left(\frac{|x|}{2}\right).
\end{aligned}
\end{equation}
\end{theorem}

\begin{IEEEproof}
See \cref{Appendix:B_PDF_CUSUM}
\end{IEEEproof}

While we derive MGF and PDF for $\lambda_k$ in above theorems, extending these proofs to arbitrary PSK is difficult, as explained in the following remark.

\begin{remark}
\label{Remark:Difficulty_in_any_MPSK}
For the QPSK to BPSK case we have $M=2$, and   
from \cref{eq:approx_llr} we have $\lambda_k = \Re\left\{(\tilde r_k)^{2}\right\}=\tilde r_{I,k}^2-\tilde r_{Q,k}^2$.
Using this, $\lambda_k$ can be expressed as sum of two independent terms, which enables the derivation of  MGF and PDF.
For an arbitrary $M$, $\lambda_k = \qllr $ will be a degree-$M$ polynomial in $\tilde r_{I,k}$ and $\tilde r_{Q,k}$ and this makes finding the closed forms for MGF or PDF, in general, a difficult problem.
%
%
\hfill $\square$
\end{remark}

\subsubsection{Analysis of $ARL2FA$ and $WADD$}
\label{subsection_Analysis_CUSUM}
We now use the analysis proposed by Wald and Siegmund~\cite[Sec.~5.2]{CDP_book_Basseville_1993} to study a relation between the given choice of threshold and the corresponding values of $ARL2FA$ and $WADD$.
%
%
%
%

Recall from \cref{Eq:Stopping_CUSUM_approx}, for the CUSUM implemented via $\lambda_k$, a bias $b$ needs to be subtracted from the CUSUM statistic. For the QPSK to BPSK setup, this bias is set to $b=\gamma_s$ (see \cref{subsection_CUSUM_algo}). Equivalently, instead of adding $\lambda_k$, we are now adding the term $\lambda_k - \gamma_s$. 
In this subsection, we denote this resulting increment by $S_k=\lambda_k-\gamma_s$. Recall that the mean of $\lambda_k$ before and after change are $0$ and $2\gamma_s$ respectively \cite[Eqs.~(29)--(32)]{PSK_classification_1995} and this implies that the means of $S_k$ before and after change-point are
\begin{equation}
    \mathbb{E}_{\scriptstyle QPSK}[S_k]=-\gamma_s,
    \qquad
    \mathbb{E}_{\scriptstyle BPSK}[S_k]=\gamma_s.
    \label{eq:Sk_means}
\end{equation}
The sequence $S_1, S_2, \ldots$ can be considered as equivalent LLRs and we consider the CUSUM procedure corresponding to it. Let $W_c^{s}$, $T_c^{s}$, and $h_c^s$ denote the corresponding CUSUM statistic, stopping time, and threshold rspectively. 
\cref{eq:Sk_means} implies that
when the transmitter is using QPSK, the CUSUM statistic $W_k^{s}$ has a negative drift and is reset to zero whenever it becomes negative, so it stays close to zero. Once the transmitter switches to BPSK, the drift becomes positive and the statistic moves towards the threshold. 
%
%
Hence, the time required to cross the threshold, and consequently both performance measures ($ARL2FA$ and $WADD$), are determined by the magnitude of this drift.
To quantify this, we use Wald's approximation to ARL of the CUSUM procedure, which is recalled in the following lemma.

\begin{lemma}[Wald's approximation for ARL]\cite[Sec.~5.2.2.2]{CDP_book_Basseville_1993}
\label{Lemma_Wald}
Let $S_1, S_2, \ldots$ denote the sequence of LLRs corresponding to the CUSUM procedure. Suppose the CUSUM statistic $W_k^{s}$ is initialized to the value $0$. Then the average time required for $W_k^{s}$ to cross the threshold $h_c^s$, denoted by $L(h_c^s)$, can be approximated as
\begin{equation}
    L(h_c^s)
    \approx
    \frac{1}{\mathbb{E}[S_k]}
    \left(
        h_c^s+\frac{e^{-\omega h_c^s}-1}{\omega}
    \right),
    \label{eq:wald_arl}
\end{equation}
where $\omega$ denotes the single non-zero root of the equation
\begin{align}
\mathbb{E}\!\left[e^{-\omega S_k}\right]=1,
\label{Eq:Root_MGF_Wald}
\end{align}
where the expectation in \cref{Eq:Root_MGF_Wald} would be taken over either the pre-change data or post-change data. 
\hfill $\square$
\end{lemma}

We now use \cref{Lemma_Wald} to study $ARL2FA$ and $WADD$ for our problem. 
We first discuss aanlysis of $WADD$ and then discuss $ARL2FA$.

\textit{Analysis of $WADD$:} 
To characterize the delay in detection,
we need to find the time required for $W_k^s$ to cross the threshold $h_c^s$ after change-point, i.e., when the data is generated by the BPSK modulation scheme. Note that due to pre-change data, $W_k^s$ might already have reached some positive value $h<h_c^s$. In this case, the delay would correspond to the time required for $W_k^s$ to cross $h_c^s-h$. 
For characterizing $ADD$, we need to consider all possible values that 
$h$ can take (and the corresponding distribution for $W_k^s$ to reach value $h$ with pre-change data) which makes the analysis of $ADD$ difficult, in general. 

In the CUSUM procedure, note that the worst delay would correspond to the situation when $W_k^s=0$ at the change-point location. Hence we can directly use \cref{Lemma_Wald} to characterize $WADD$. Let $\omega_1$ denote the non-zero root of \cref{Eq:Root_MGF_Wald} when the expectation is taken under BPSK, i.e., $\mathbb{E}_{\scriptstyle BPSK}[e^{-\omega_1 S_k}]=1$. 
Since MGF of $S_k$ is defined as $\mathbb{E}_{\scriptstyle BPSK}[e^{t S_k}]$, we have $t=-\omega_1$.
We can find $\omega_1$ by using the MGF of $\lambda_k$ for BPSK derived in \cref{Thm:MGF_QPSK_BPSK}. 
Since $S_k=\lambda_k-\gamma_s$, we have $\mathbb{E}_{\scriptstyle BPSK}[e^{-\omega_1 S_k}]=e^{\gamma_s\omega_1}M_{\lambda}(-\omega_1\mid\mathrm{BPSK})$. Substituting the MGF from \cref{Thm:MGF_QPSK_BPSK} into \cref{Eq:Root_MGF_Wald}, $\omega_1$ is the solution of
\begin{equation}
    \frac{1}{\sqrt{1-4\omega_1^2}}
    \exp\!\left(\gamma_s\omega_1-\frac{2\gamma_s\omega_1}{1+2\omega_1}\right)=1,    
    \label{eq:omega1}
\end{equation}
where the range of $\omega_1$ is the region where $|-\omega_1|<1/2$. 
Since $\mathbb{E}_{\scriptstyle BPSK}[S_k]=\gamma_s>0$ we further get a condition that $\omega_1>0$.
In general, we observe that 
finding the closed-form solution \cref{eq:omega1} for $\omega_1$ is not possible. However, $\omega_1$ satisfying \cref{eq:omega1} can be obtained numerically by searching over $0<\omega_1<1/2$.

\textit{Analysis of $ARL2FA$:}
To characterize false alarm, we need the find time required for $W_k^s$ to cross the threshold $h_c^s$ when the data is generated by QPSK, i.e., before change-point. In this case, since $W_k^s$ starts at zero 
we can use \cref{Lemma_Wald} to approximate $ARL2FA$, with the expectation in \cref{eq:wald_arl} and \cref{Eq:Root_MGF_Wald} taken under QPSK.
Let $\omega_0$ denote the non-zero root of \cref{Eq:Root_MGF_Wald} when the expectation is taken under QPSK, i.e., $\mathbb{E}_{\scriptstyle QPSK}[e^{-\omega_0 S_k}]=1$. Using $S_k=\lambda_k-\gamma_s$ and the MGF of $\lambda_k$ for QPSK from \cref{Thm:MGF_QPSK_BPSK}, we have $\mathbb{E}_{\scriptstyle QPSK}[e^{-\omega_0 S_k}]=e^{\gamma_s\omega_0}M_{\lambda}(-\omega_0\mid\mathrm{QPSK})$. Since $\mathbb{E}_{\scriptstyle QPSK}[S_k]=-\gamma_s<0$, the root now satisfies $\omega_0<0$, and $\omega_0$ is the solution of
\begin{multline}
    \frac{e^{\gamma_s\omega_0}}{2\sqrt{1-4\omega_0^2}}
    \Bigg[
        \exp\!\left(-\frac{2\gamma_s\omega_0}{1+2\omega_0}\right)
        +\exp\!\left(\frac{2\gamma_s\omega_0}{1-2\omega_0}\right)
    \Bigg]=1,
    \label{eq:omega0}
\end{multline}
in the region $|\omega_0|<1/2$. As for $\omega_1$, this equation has no closed-form solution, and $\omega_0$ is obtained numerically by searching over $-1/2<\omega_0<0$.

We observe that $ARL2FA$ obtained from \cref{Lemma_Wald} is not accurate for our setup (see \cref{Sec:Simulations}).
Similar observations are also registered by Siegmund~\cite[Sec.~5.2.2.3]{CDP_book_Basseville_1993} and has suggested a method to rectify it. Siegmund proposed to replace $h_c^s$ in \cref{eq:wald_arl} by a corrected threshold $h_c^{*}$ defined as
\begin{equation}
    h_c^{*}=h_c^{s}+\varrho_+-\varrho_-,
    \label{eq:hstar}
\end{equation}
where $\varrho_+$ and $\varrho_-$ denote the expected overshoot of the cumulative sum of $S_k$ over the upper boundary $h_c^s$ and over the lower boundary $0$, respectively. Note that $\varrho_+$ and $\varrho_-$ are constants that depend only on the distribution of $S_k$ under QPSK, and not on the threshold. We compute them using the method described in \cite[Sec.~5.2]{CDP_book_Basseville_1993}.
\subsection{Analysis of MOSUM for any $M^{\prime}$-PSK to $M$-PSK}
\label{Sec:PSK}


For the MOSUM method, finding the analogous approximations for $ARL2FA$ and $WADD$ are not straightforward~\cite{CDP_book_Basseville_1993} and hence we use the bounds provided by Bohm--Hackl~\cite{bohm1990improved}. For these bounds, we need to find the PDF of the vector corresponding to the MOSUM statistic.

For $k=w, w+1, \ldots$, the MOSUM statistic $Y_k^\lambda$ is given by
\begin{equation}
    Y^\lambda_k
    =
    \sum_{k^{\prime}=k-w+1}^{k}\lambda_{k^{\prime}}.    
    \label{Eq:MOSUM_stat_lambda}
\end{equation}
Let us define a vector $\mathbf{Y}_n$ corresponding to the sequence of MOSUM statistics $Y^\lambda_k$, for $n=1, 2, \ldots$ as follows 
\begin{equation}
    \mathbf{Y}_n
    =
    \begin{bmatrix}
        Y_{w}^\lambda &
        Y_{w+1}^\lambda &
        \cdots &
        Y_{w+n-1}^\lambda
    \end{bmatrix}^{T}.
    \label{Eq:MOSUM_stat_lambda_vec}
\end{equation}
We now consider the situation when 
all $\lambda_{k^{\prime}}$ \cref{Eq:MOSUM_stat_lambda} corresponds to the data generated by the same modulation scheme. It can be seen from \cref{Eq:MOSUM_stat_lambda}, for a sufficiently large $w$, $Y_k^\lambda$ can be approximated by a Gaussian random variable. However, 
any $Y_k^\lambda$ and $Y_{k+j}^\lambda$, with $\lvert d \rvert \leq w-1$ will depend on each other. Thus, in the following theorem we approximate the distribution of $\mathbf{Y}_n$ by a Gaussian distribution.

\begin{theorem}[Gaussian approximation of the MOSUM vector]
\label{thm:Gaussian_Approx}
For sufficiently large $w$, the joint distribution of $\mathbf{Y}_n$ defined in \cref{Eq:MOSUM_stat_lambda_vec} can be approximated by the multivariate Gaussian distribution as follows
\begin{equation}
    \mathbf{Y}_n
    \overset{\mathrm{approx.}}{\sim}
    \mathcal{N}
    \left(
        w\mu\mathbf{1}_n,
        \boldsymbol{\Sigma}_n
    \right),
\end{equation}
$\mathbf{1}_n$ denotes the all-ones vector and entry $\left[\boldsymbol{\Sigma}_n\right]_{i,j}$ in the covariance matrix $\boldsymbol{\Sigma}_n$ corresponding to the $i$-th row and $j$-th column is given by
\begin{equation}
    \left[\boldsymbol{\Sigma}_n\right]_{i,j}
    =
    \sigma^2
    \max\!\left(w-\lvert i-j\rvert,0\right),
    \qquad
    1\leq i,j\leq n,
\end{equation}
where $\mu$ and $\sigma^2$ are the mean and variance of $\lambda_{k^{\prime}}$ in \cref{Eq:MOSUM_stat_lambda}.
%
%
\end{theorem}
\begin{IEEEproof}
The proof is straightforward from the central limit theorem. For the sake of completeness we provide it in~\cref{Appendix:mosum_approximation_proof}. 
\end{IEEEproof}

The mean and variance of $\lambda_k$ was derived in \cite{PSK_classification_1995}, which are recalled in \cref{lem:Moments_of_LLR}. Note that \cref{lem:Moments_of_LLR} applies to arbitrary PSK modulation schemes, thereby enabling the analysis of MOSUM for arbitrary PSK.

\begin{lemma}[Mean and variance of $\lambda_k$]\cite[Appendix~A]{PSK_classification_1995}
\label{lem:Moments_of_LLR}
Consider $\lambda_k$ is defined in \cref{eq:approx_llr} with $M^{\prime} \geq 2M$. For the given SNR $\gamma_s$, define a constant $V_M$ as follows
\begin{equation}
    V_{\scriptstyle M}
    =
    \sum_{l=0}^{M}
    \frac{(M!)^2}
    {2^l[(M-l)!]^2}
    \gamma_s^{-l}.
\end{equation}
When data is generated by $M$-PSK, the mean and variance of $\lambda_k$, denoted by $\mu_{M}$ and $\sigma_M^2$ respectively are given by
\begin{equation}
    \mu_{\scriptstyle M}
    =
    (2\gamma_s)^{M/2},
    \qquad
    \sigma_{\scriptstyle M}^2
    =
    (2\gamma_s)^M V_{\scriptstyle M}.
\end{equation}
When data is generated by $M^{\prime}$-PSK, the mean and variance of $\lambda_k$, denoted by $\mu_{M^{\prime}}$ and $\sigma_{M^{\prime}}^2$ respectively are given by
\begin{equation}
\begin{aligned}
    \mu_{\scriptstyle M^{\prime}}
    &=
    0,
    \\[6pt]
    \sigma_{\scriptstyle M^{\prime}}^2
    &=
    \begin{cases}
        (2\gamma_s)^M(V_{\scriptstyle M}+1),
        & \mbox{~if~}M^{\prime}=2M,\\[4pt]
        (2\gamma_s)^M\left(V_{\scriptstyle M}+0.5\right),
        & \mbox{~if~} M^{\prime}>2M.
    \end{cases}
\end{aligned}
\end{equation}
\end{lemma}

\subsubsection{Performance analysis of MOSUM via ARL2FA and WADD}
\label{subsection_Analysis_MOSUM}

To derive bounds on $ARL2FA$ and $WADD$ for 
the MOSUM method, we use the results derived by Bohm--Hackl in \cite{bohm1990improved}, recalled in \cref{lem:Bohm_Hackl_MOSUM}.

\begin{lemma}[Bounds for MOSUM]\cite{bohm1990improved}
\label{lem:Bohm_Hackl_MOSUM}
Consider an i.i.d.~sequence $\lambda_1, \lambda_2, \ldots$ with a finite constant mean. For the fixed window length $w$, define the moving sum $Y_k^{\lambda}$ at
instant $k$ as $Y_k^{\lambda} = \sum_{k^{\prime}=k-w+1}^k{\lambda_{k^{\prime}}}$, where $k=w, w+1, \ldots$. 
%
%
For the given positive constant $h$, define $\zeta_j(h)$ and $\rho_j(h)$ as follows:
\begin{align*}
    \zeta_j(h) 
    &=\Pr\!\left(
    Y_w^{\lambda}<h,\,
    Y_{w+1}^{\lambda}<h,\,
    \ldots,\,
    Y_{w+j-1}^{\lambda}<h
    \right)\\
    \rho_j(h)
    &=
    \Pr\!\left(
    Y_w^{\lambda}<h,\,
    \ldots,\,
    Y_{w+j-2}^{\lambda}<h,\,
    Y_{w+j-1}^{\lambda}\geq h
    \right)
\end{align*}
where note that $\rho_j(h)=\zeta_{j-1}(h)-\zeta_j(h)$ and $j=0, 1, \ldots$ with $\zeta_0(h)=1$. 
%
%
Let $T_m^{a}$ be the stopping time for the threshold $h$. Then the average run length $\mathbb{E}[T_m^a]$ of the corresponding MOSUM process satisfies the bounds:
\begin{equation}
    \frac{w}
        {1-\zeta_w(h)}
    \leq
    \mathbb{E}[T_m^a]
    \leq
    w+
        \frac{
            \zeta_w(h)}{
            \rho_w(h)
        }
\end{equation}
\end{lemma}

We now use \cref{lem:Bohm_Hackl_MOSUM} to study the bounds on $ARL2FA$ and $WADD$ for our CPD problem when modulation scheme switched from $M^{\prime}$-PSK to $M$-PSK. For $ARL2FA$, we need to focus on the situation when the data is generated by $M^{\prime}$-PSK. Since in this case all data samples are generated by the same modulation scheme, we can use \cref{lem:Bohm_Hackl_MOSUM} to get bounds on $ARL2FA$. In \cref{thm:Gaussian_Approx} we now substitute $\mu=\mu_{\scriptstyle M^{\prime}}$ and $\sigma=\sigma_{\scriptstyle M^{\prime}}$ from \cref{lem:Moments_of_LLR} to get the distribution of $\mathbf{Y}_n$, which is required to find $\zeta_j(h)$ and $\rho_j(h)$. For $WADD$, we apply \cref{lem:Bohm_Hackl_MOSUM} to the post-change MOSUM process. Following the change to $M$-PSK, the first $w-1$ moving windows contain observations from both modulation schemes. Assuming that these initial windows do not trigger an alarm, the subsequent MOSUM statistics are formed entirely from post-change observations and follow the process considered in \cref{lem:Bohm_Hackl_MOSUM}. We therefore use $\mu=\mu_{\scriptstyle M}$ and $\sigma=\sigma_{\scriptstyle M}$ to evaluate \(\zeta_w(h_m^\lambda)\) and \(\rho_w(h_m^\lambda)\), and apply the bounds in \cref{lem:Bohm_Hackl_MOSUM} to the resulting expected run length.

\section{Simulations}
\label{Sec:Simulations}

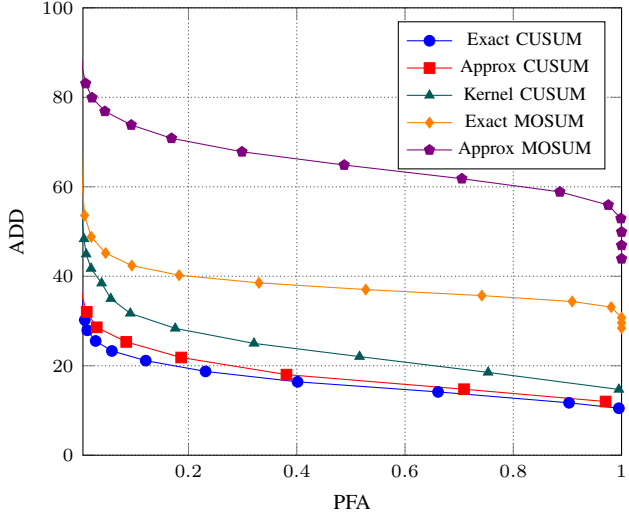
\begin{figure}[!t]
\centering
\begin{tikzpicture}
\begin{axis}[
    xlabel={PFA},
    ylabel={ADD},
    xmin=0.005, xmax=1,
    ymin=0, ymax = 100,
    grid=both,
    width=\columnwidth,
    height=7.5cm,
    legend style={font=\scriptsize},
    tick label style={font=\scriptsize},
    label style={font=\footnotesize},
    legend pos=north east,
]

\addplot[blue, mark=*] coordinates {
(0.995,10.5) (0.903,11.7165) (0.661,14.1519)
(0.4015,16.406) (0.2315,18.73) (0.121,21.1428)
(0.0585,23.3022) (0.0285,25.5389) (0.013,27.926)
(0.0085,30.2295) (0.004,32.6135) (0.0035,34.8801)
};
\addlegendentry{Exact CUSUM}

\addplot[red, mark=square*] coordinates {
(0.9705,12) (0.709,14.7526) (0.381,18.017)
(0.1865,21.8371) (0.085,25.3279) (0.0305,28.5869)
(0.012,32.0369) (0.004,35.3017) (0.0005,38.968)
};
\addlegendentry{Approx CUSUM}

\addplot[teal!70!black, mark=triangle*] coordinates {
(0.995,14.7) (0.7535,18.5193)
(0.516,22.0424) (0.321,24.9941) (0.1755,28.3572)
(0.0925,31.7118) (0.0565,34.9968) (0.0395,38.4524)
(0.02,41.7286) (0.011,44.9414) (0.007,48.3449)
(0.004,51.7098) (0.0025,55.1128)
};
\addlegendentry{Kernel CUSUM}

\addplot[orange, mark=diamond*] coordinates {
(1,28.4199) (1,29.5978) (1,30.747) (0.9809,33.1036) (0.9087,34.3579)
(0.7419,35.6808) (0.5274,37.0239) (0.3302,38.5188)
(0.1827,40.2318) (0.0954,42.3877) (0.0469,45.1887)
(0.0204,48.8043) (0.0081,53.5993) (0.0037,60.6906)
(0.0016,71.1866) (0.0005,87.0419) (0.0003,113)
(0.0002,160)
};
\addlegendentry{Exact MOSUM}

\addplot[violet, mark=pentagon*] coordinates {
(1,43.912) (1,46.915) (1,49.8982)
(0.9987,52.9125) (0.9755,55.9233) (0.886,58.868)
(0.705,61.808) (0.4876,64.8675) (0.2988,67.8269)
(0.1686,70.8404) (0.0943,73.8177) (0.0456,76.8645)
(0.0218,79.9143) (0.0101,83.0931) (0.0043,86.3129)
(0.0018,89.8417)
};
\addlegendentry{Approx MOSUM}


\end{axis}

\end{tikzpicture}
\caption{ADD versus PFA for exact and approximate LLR-based CUSUM and MOSUM detectors, and the kernel-based CUSUM benchmark, for the QPSK-to-BPSK transition at SNR = 2 dB. Results use \(10^3\) pre-change samples and \(10^5\) post-change samples. The MOSUM window lengths are \(w=45\) for the exact LLR and \(w=110\) for the approximate LLR.}

    \label{fig:PFA_ADD_QPSK_BPSK}
\end{figure}


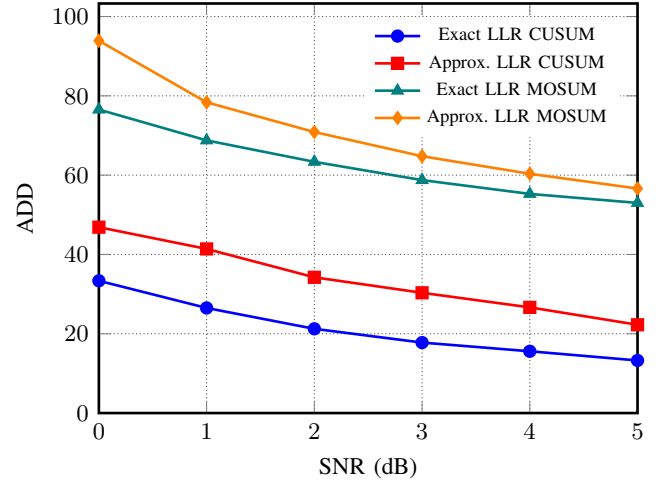
\begin{figure}[!t]
    \centering
    \begin{tikzpicture}
\begin{axis}[
    width=\columnwidth,
    height=7cm,
    xlabel={SNR (dB)},
    ylabel={ADD},
    xmin=0, xmax=5,
    ymin=0,
    xtick={0,1,2,3,4,5},
    grid=major,
    legend pos=north east,
    legend style={
        font=\scriptsize,
        draw=none,
        fill=white
    },
    tick label style={font=\small},
    label style={font=\small},
    line width=1pt,
    mark size=2pt,
]

\addplot[
    blue,
    solid,
    mark=*,
]
coordinates {
    (0,33.3571)
    (1,26.505)
    (2,21.2467)
    (3,17.774)
    (4,15.5835)
    (5,13.2622)
};
\addlegendentry{Exact LLR CUSUM}

\addplot[
    red,
    solid,
    mark=square*,
]
coordinates {
    (0,46.8764)
    (1,41.3961)
    (2,34.2364)
    (3,30.3308)
    (4,26.6559)
    (5,22.2661)
};
\addlegendentry{Approx.\ LLR CUSUM}

\addplot[
    teal,
    solid,
    mark=triangle*,
]
coordinates {
    (0,76.54)
    (1,68.77)
    (2,63.36)
    (3,58.75)
    (4,55.27)
    (5,53.00)
};
\addlegendentry{Exact LLR MOSUM}

\addplot[
    orange,
    solid,
    mark=diamond*,
]
coordinates {
    (0,93.903)
    (1,78.37)
    (2,70.88)
    (3,64.775)
    (4,60.35)
    (5,56.62)
};
\addlegendentry{Approx.\ LLR MOSUM}

\end{axis}
\end{tikzpicture}
    \caption{ADD versus SNR for exact and approximate LLR-based CUSUM and MOSUM detectors for the QPSK-to-BPSK transition. The results illustrate the effect of SNR on detection delay and compare the performance of the exact and approximate LLR implementations.}
    \label{fig:add_snr_16psk_8psk}
\end{figure}

In this section, we study the performance of the CUSUM and MOSUM
procedures described in~\cref{Sec:Proposed_method} using simulations. We compare the different methods and verify that the results of ~\cref{Sec:QPSK_BPSK} and ~\cref{Sec:PSK} can be used to determine the thresholds $h_c$ and $h_m$.

To generate the data, we recall from ~\cref{Sec:System_model} that all the statistics used by our methods, namely $\ell_k$ and $\lambda_k$, are functions of the received waveform $r_k(t)$ only through the normalized observation $\tilde{r}_k$ of \cref{Eq:r_k_tilda}. Hence, simulating the passband waveforms is not required, and it suffices to generate $\tilde{r}_k$ with its exact distribution. For the $k$-th symbol, we first draw the phase $\theta_k$ uniformly from the $M^{\prime}$ phases of $M^{\prime}$-PSK if $k<\tau$ and from the $M$ phases of $M$-PSK if $k\geq\tau$ (see Eq.~(2)), independently across $k$. We then generate 
\begin{equation}     
\tilde{r}_k = \sqrt{2\gamma_s}\, e^{j\theta_k} + n_k, 
\end{equation}
where the real and imaginary parts of $n_k$ are independent $\mathcal{N}(0,1)$ random variables, independent across $k$. The amplitude $\sqrt{2\gamma_s}$ follows from the normalization in \cref{Eq:r_k_tilda} with $\gamma_s$ denoting the symbol SNR. This construction preserves the exact distribution of the normalized observations under both the pre-change and post-change hypotheses, while avoiding unnecessary passband simulation. The generated observations are then used to compute the corresponding statistics and evaluate the detection performance of each method.

\begin{figure*}[!t]
    \centering

    \begin{minipage}{0.95\textwidth}
        \centering
        \begin{minipage}{0.48\textwidth}
            \centering

\begin{tikzpicture}
\begin{semilogyaxis}[
    width=\columnwidth,
    height=6.0cm,
    xlabel={Threshold $h$},
    ylabel={ARL2FA},
    title={QPSK to BPSK CUSUM, SNR = 2 dB},
    grid=major,
    legend pos=north west,
    legend style={
        font=\scriptsize,
        draw=black,
        fill=white,
    },
    tick label style={font=\scriptsize},
    label style={font=\scriptsize},
    title style={font=\footnotesize},
]

\addplot[
    red,
    only marks,
    mark=*,
    mark size=1.4pt,
]
table[
    col sep=comma,
    x=threshold,
    y=arl_sim
] {Figures/cusum/ARL2FA_results.csv};
\addlegendentry{CUSUM Simulation}

\addplot[
    blue,
    dashed,
    line width=1.2pt,
]
table[
    col sep=comma,
    x=threshold,
    y=arl_siegmund
] {Figures/cusum/ARL2FA_results.csv};
\addlegendentry{Siegmund}

\addplot[
    yellow!70!black,
    dashed,
    line width=1.2pt,
]
table[
    col sep=comma,
    x=threshold,
    y=arl_wald
] {Figures/cusum/ARL2FA_results.csv};
\addlegendentry{Wald}

\end{semilogyaxis}
\end{tikzpicture}
        \end{minipage}
        \hfill
        \begin{minipage}{0.48\textwidth}
            \centering

\begin{tikzpicture}
\begin{axis}[
    width=\columnwidth,
    height=6.0cm,
    xlabel={Threshold $h$},
    ylabel={WADD},
    title={QPSK to BPSK CUSUM, SNR = 2 dB},
    grid=major,
    legend pos=north west,
    legend style={
        font=\scriptsize,
        draw=black,
        fill=white,
    },
    tick label style={font=\scriptsize},
    label style={font=\scriptsize},
    title style={font=\footnotesize},
]

\addplot[
    red,
    only marks,
    mark=*,
    mark size=1.4pt,
]
table[
    col sep=comma,
    x=threshold,
    y=wadd_sim
] {Figures/cusum/WADD_results.csv};
\addlegendentry{CUSUM Simulation}

\addplot[
    blue,
    dashed,
    line width=1.2pt,
]
table[
    col sep=comma,
    x=threshold,
    y=wadd_siegmund
] {Figures/cusum/WADD_results.csv};
\addlegendentry{Siegmund}

\addplot[
    yellow!70!black,
    dashed,
    line width=1.2pt,
]
table[
    col sep=comma,
    x=threshold,
    y=wadd_wald
] {Figures/cusum/WADD_results.csv};
\addlegendentry{Wald}

\end{axis}
\end{tikzpicture}
        \end{minipage}

        \captionof{figure}{$ARL2FA$ and $WADD$ versus threshold
        for the QPSK-to-BPSK CUSUM at SNR = 2~dB.}
        \label{fig:Bounds_CUSUM}
    \end{minipage}

    \vspace{0.8cm}

    \begin{minipage}{0.95\textwidth}
        \centering
        \begin{minipage}{0.48\textwidth}
            \centering
            \begin{tikzpicture}
\begin{axis}[
    xlabel={Threshold($h_m^\lambda$)},
    ylabel={$ARL2FA$},
    ymode=log,
    xmin=5, xmax=185,
    ymin=90, ymax = 100000,
    grid=both,
    width=\columnwidth,
    height=7.5cm,
    legend style={font=\scriptsize},
    tick label style={font=\scriptsize},
    label style={font=\footnotesize},
    legend pos=north west,
]

\addplot[blue, mark=*] coordinates {
    (10,145.080)
    (20,161.266)
    (30,185.514)
    (40,216.550)
    (50,262.806)
    (60,324.910)
    (70,419.512)
    (80,551.156)
    (90,751.909)
    (100,1074.023)
    (110,1612.008)
    (120,2429.037)
    (130,3785.145)
    (140,6339.220)
    (150,10961.355)
    (160,19699.303)
    (170,36982.637)
    (180,71707.822)
};
\addlegendentry{$ARL2FA_{MC}$}

\addplot[red, mark=square*] coordinates {
    (10,119.241)
    (20,129.642)
    (30,144.661)
    (40,166.217)
    (50,197.277)
    (60,242.493)
    (70,309.256)
    (80,410.040)
    (90,565.702)
    (100,812.519)
    (110,1217.668)
    (120,1901.621)
    (130,3098.049)
    (140,5268.461)
    (150,9363.120)
    (160,17695.224)
    (170,33830.450)
    (180,67845.343)
};
\addlegendentry{$L_{\scriptscriptstyle{M^{\prime}}}$}

\addplot[black, mark=triangle*] coordinates {
    (10,175.174)
    (20,193.304)
    (30,217.007)
    (40,249.123)
    (50,293.270)
    (60,355.026)
    (70,443.513)
    (80,573.699)
    (90,770.584)
    (100,1077.488)
    (110,1570.898)
    (120,2392.125)
    (130,3806.865)
    (140,6344.837)
    (150,11040.595)
    (160,20117.670)
    (170,38327.988)
    (180,76417.255)
};
\addlegendentry{$U_{\scriptscriptstyle{M^\prime}}$}

\end{axis}

\end{tikzpicture}
\caption*{ARL2FA Bounds for $w$ = 100}
        \end{minipage}
        \hfill
        \begin{minipage}{0.48\textwidth}
            \centering
            \begin{tikzpicture}
\begin{axis}[
    xlabel={Threshold($h_m^\lambda$)},
    ylabel={WADD},
    ymode=log,
    xmin=285, xmax=445,
    ymin=90, ymax=2900,
    grid=both,
    width=\columnwidth,
    height=7.5cm,
    legend style={font=\scriptsize},
    tick label style={font=\scriptsize},
    label style={font=\footnotesize},
    legend pos=north west,
]

\addplot[blue, mark=*] coordinates {
    (290,104.369)
    (300,107.024)
    (310,111.409)
    (320,118.367)
    (330,128.369)
    (340,142.205)
    (350,161.840)
    (360,188.362)
    (370,227.622)
    (380,283.538)
    (390,362.901)
    (400,480.181)
    (410,655.332)
    (420,945.955)
    (430,1402.242)
    (440,2160.558)
};
\addlegendentry{$WADD_{MC}$}

\addplot[red, mark=square*] coordinates {
    (290,100.692)
    (300,101.513)
    (310,103.069)
    (320,105.851)
    (330,110.465)
    (340,117.876)
    (350,129.454)
    (360,146.732)
    (370,173.169)
    (380,212.612)
    (390,274.529)
    (400,372.668)
    (410,532.662)
    (420,800.764)
    (430,1267.971)
    (440,2133.668)
};
\addlegendentry{$L_{\scriptscriptstyle{M}}$}

\addplot[black, mark=triangle*] coordinates {
    (290,126.009)
    (300,131.494)
    (310,138.313)
    (320,146.676)
    (330,157.321)
    (340,172.182)
    (350,192.959)
    (360,220.261)
    (370,259.568)
    (380,316.685)
    (390,397.285)
    (400,524.857)
    (410,733.124)
    (420,1067.550)
    (430,1636.392)
    (440,2665.724)
};
\addlegendentry{$U_{\scriptscriptstyle{M}}$}

\end{axis}
\end{tikzpicture}
\caption*{WADD Bounds for $w$ = 100}
        \end{minipage}

        \captionof{figure}{Theoretical bounds and empirical values of
        $ARL2FA$ and $WADD$ for MOSUM with $w=100$, for the
        8-PSK-to-BPSK transition at SNR = 2~dB.}
        \label{fig:Bounds_MOSUM}
    \end{minipage}

\end{figure*}

The measures of \cref{subsec:performance_measures} are estimated using $10^5$ simulation trials. $ARL2FA$ is estimated by generating $M^{\prime}$-PSK data only and averaging the stopping times. The $WADD$ is estimated by generating $M$-PSK data only, starting the detector from its initial state (CUSUM statistic equal to zero, MOSUM window containing only $M$-PSK observations), which is the worst case of~\cref{Sec:QPSK_BPSK,Sec:PSK}. The $ADD$ is estimated with $\tau$ fixed, using $10^3$ pre-change and $10^5$ post-change symbols, and averaged over trials with $T\geq\tau$. We also record the probability of false alarm, $PFA=\mathbb{P}_{\tau}[T<\tau]$, which denotes the probability that the detector raises an alarm before the actual change point $\tau$, under the pre-change distribution.

\textit{CUSUM:} We consider the QPSK-to-BPSK transition of~\cref{Sec:QPSK_BPSK} at $2$~dB. Since $M=2$, the bias correction is $b=\gamma_s$, so $W_k^{\lambda}$ coincides with $W_k^{s}$ and $h_c^{s}=h_c^{\lambda}$. We compare the theoretical approximations of $ARL2FA$ and $WADD$ with empirical values in \cref{fig:Bounds_CUSUM}.

\subsection{Simulation results for the proposed CPD methods}

We compare the following five methods in terms of $PFA$ and $ADD$. 

\begin{itemize}
\item \emph{Exact CUSUM:} Uses the exact log-likelihood ratio $\ell_k$ to recursively update the cumulative sum, raising an alarm when it exceeds the threshold.

\item \emph{Approx CUSUM:} Uses the approximate log-likelihood ratio $\lambda_k$ to form the CUSUM increment $S_k=\lambda_k-b$, (here $b=\tfrac{1}{2}(2\gamma_s)^{M/2}$ is the bias correction explained in \cref{Eq:Stopping_CUSUM_approx}).

\item \emph{Kernel CUSUM:} A non-parametric, kernel-based change-point detector \cite{flynn2019change}, included as a benchmark.

\item \emph{Exact MOSUM:} Uses $\ell_k$ over a fixed-length sliding window, detecting changes based on recent observations.

\item \emph{Approx MOSUM:} Applies the MOSUM framework to $\lambda_k$, using the statistic $Y_k^{\lambda}$, threshold $h_m^{\lambda}$, and stopping time $T_m^{\lambda}$.
\end{itemize}



As shown in \cref{fig:PFA_ADD_QPSK_BPSK}, Exact CUSUM achieves the lowest $ADD$ across the plotted range, closely followed by Approx CUSUM. Kernel CUSUM exhibits consistently higher detection delays, while both MOSUM variants perform worse, particularly Approx. MOSUM. The gap between the exact and approximate CUSUM methods remains small, indicating that the approximation retains most of the detection performance. In contrast, the higher delays of MOSUM suggest that restricting accumulation to a finite window comes at a cost in detection speed for this transition.

 \cref{fig:add_snr_16psk_8psk} compares the detection delays of the five methods across SNR values for 16-PSK-to-8-PSK transition. The reduction in $ADD$ with increasing SNR indicates that improved signal quality facilitates faster detection. The relative gaps between the curves show how closely the approximate methods follow their exact counterparts and the performance differences between methods persist across SNR levels.

\subsection{Verification of the analytical results}
\label{subsec:verification_analytical_results}


We can observe in \cref{fig:Bounds_CUSUM} that $ARL2FA$ grows exponentially with the threshold, whereas $WADD$ grows approximately linearly, consistent with the theoretical analysis in \cref{subsection_Analysis_CUSUM}. Wald's approximation substantially underestimates $ARL2FA$, while the correction in \cref{eq:hstar} brings the approximation into close agreement with simulations. For $WADD$, both Wald and Siegmund approximations agree well with simulations at large thresholds, consistent with their described behavior. These results show that the derived approximations can guide threshold selection to meet a prescribed $ARL2FA$ requirement while keeping the $WADD$ within a desired limit.


\textit{MOSUM:} We verify the bounds of \cref{subsection_Analysis_MOSUM} for the 8-PSK-to-BPSK transition ($M^{\prime}=8$, $M=2$) with $w=100$ at $2$~dB (see \cref{fig:Bounds_MOSUM}). Since $M^{\prime}>2M$, we use the second case of \cref{lem:Moments_of_LLR} to compute $\sigma_{M^{\prime}}^2$. 

In \cref{fig:Bounds_MOSUM}, the simulated $ARL2FA$ and $WADD$ increase with the threshold and are closely tracked by their respective bounds. The $WADD$ estimates remain within the bounds across the tested range, while $ARL2FA$ shows small deviations at a few thresholds due to the Gaussian approximation in \cref{thm:Gaussian_Approx}. The results suggest that the derived bounds provide useful estimates of the false-alarm rate and detection delay for MOSUM.

\section{Conclusion and future work}
\label{Sec:Conclusion}


In this work, we studied sequential change-point detection (CPD) for a switch in the modulation scheme from $M^{\prime}$-PSK to $M$-PSK over an AWGN channel. We proposed CUSUM and MOSUM detectors, each implemented using either the exact LLRs $\ell_k$ or the approximate LLRs $\lambda_k$. For CUSUM, we derived the MGF and PDF of $\lambda_k$ for the QPSK-to-BPSK case (Theorems~\ref{Thm:MGF_QPSK_BPSK} and~\ref{Thm:PDF_QPSK_BPSK}) and used them to characterize ARL2FA and WADD. For MOSUM, we obtained bounds on both measures for any $M^{\prime}$-PSK to $M$-PSK transition (\cref{thm:Gaussian_Approx}). The simulations in \cref{Sec:Simulations} show that the approximate CUSUM performs close to the exact CUSUM, and that the analytical bounds accurately predict the empirical behavior of the MOSUM.

Future work will focus on extending the CUSUM analysis to arbitrary $M^{\prime}$-PSK to $M$-PSK transitions, for which closed-form expressions for the MGF and PDF of $\lambda_k$ remain difficult to obtain (see \cref{Remark:Difficulty_in_any_MPSK}). We also plan to study the offline setting, where the entire observed sequence is available before a decision is made, and non-coherent, asynchronous channels in which the carrier phase and symbol timing are unknown.

\appendices   
\crefalias{section}{appendix}
\crefname{appendix}{Appendix}{Appendices}
\Crefname{appendix}{Appendix}{Appendices}

\section{Proof of \cref{Thm:MGF_QPSK_BPSK}}
\label{Appendix:A_MGF_CUSUM}

In this appendix we prove \cref{Thm:MGF_QPSK_BPSK}. Toward this, we first derive the conditional MGF of $\lambda_k$ when the phase of the transmitted symbol $\theta_k$ is given in \cref{Proposition:MGF_QPSK_BPSK}. Using this proposition we then derive the MGF of $\lambda_k$ when the data is generated by QPSK and BPSK.

\begin{proposition}
\label{Proposition:MGF_QPSK_BPSK}
Let $\theta_k$ be the phase of the transmitted symbol. Then the conditional MGF of the approximate LLR $\lambda_k$ defined in \cref{eq:approx_llr} is given by
\begin{align}    
    M_{\lambda_k}(t\mid\theta_k)
    =
    \frac{1}{\sqrt{1-4t^2}}
    \exp\!\left\{
    2\gamma_s t\,
    \left[
    \frac{\cos^2\theta_k}{1-2t}
    -
    \frac{\sin^2\theta_k}{1+2t}
    \right]
    \right\}, 
    \label{eqn:MGF_QPSK_BPSK}
\end{align}
where $|t|<1/2$.
\hfill $\square$
\end{proposition}
\begin{IEEEproof}
The key idea is to write $\lambda_k$ as the difference of squares of two independent Gaussian random variables, evaluate the MGF of each square using the noncentral chi-square MGF, and then finally obtain the required MGF using the properties of MGFs of independent random variables.

We begin by simplifying $\lambda_k$. 
Recall that (see \cref{eq:approx_llr}), for QPSK to BPSK we have $M=2$ and $\lambda_k$ is given by
\begin{equation}
\lambda_k = \Re\{(\tilde r_k)^2\}.
\label{Eq:Approx_LLR_QPSK_BPSK_appendix}
\end{equation}
From \cref{Eq:r_k_tilda_real_imag}, $\tilde r_k=\tilde r_{I,k}+j\tilde r_{Q,k}$ and substituting this in \cref{Eq:Approx_LLR_QPSK_BPSK_appendix} we get
\begin{equation}
\begin{aligned}
\lambda_k 
&= \Re\left\{ (\tilde r_{I,k} + j\tilde r_{Q,k})^2 \right\} \\
&= \Re\left\{ \left(\tilde r_{I,k}^2 - \tilde r_{Q,k}^2\right) + j2\tilde r_{I,k}\tilde r_{Q,k} \right\}\\
&=\tilde r_{I,k}^2-\tilde r_{Q,k}^2.
\label{Eq:lambda_k_appendix}
\end{aligned}
\end{equation}

We next find the distributions of $\tilde r_{I,k}$ and $\tilde r_{Q,k}$. Recall from \cref{Eq:r_k_tilda_real_imag} that $\tilde r_k$ is the normalized observation of \cref{Eq:r_k_tilda}. Conditioned on $\theta_k$, it is the sum of a deterministic symbol term and a noise term, 
$
\tilde r_k=\sqrt{2\gamma_s}\,e^{j\theta_k}+n_k,$
where $n_k=n_{I,k}+jn_{Q,k}$ is the normalized noise sample. As stated in \cref{Sec:System_model}, real and imaginary parts of $n_k$ (represented by $n_{I,k}$ and $n_{Q,k}$ respectively) are independent zero mean unit variance Gaussian random variables. Taking the real and imaginary parts, and using $\Re\{e^{j\theta_k}\}=\cos\theta_k$ and $\Im\{e^{j\theta_k}\}=\sin\theta_k$, we get
\begin{align}
\label{Eq:r_I}
\tilde r_{I,k}&=\sqrt{2\gamma_s}\cos\theta_k+n_{I,k}\\
\label{Eq:r_Q}
\tilde r_{Q,k}&=\sqrt{2\gamma_s}\sin\theta_k+n_{Q,k}.
\end{align}
Since $\theta_k$ is given and $\gamma_s$ is the constant SNR, the first term \cref{Eq:r_I} and \cref{Eq:r_Q} is a constant, the distributions of $\tilde r_{I,k}$ and $\tilde r_{Q,k}$ will be
\begin{align}
    \tilde r_{I,k} &\sim\mathcal N\!\big(\sqrt{2\gamma_s}\cos\theta_k,\,1\big),\\
    \tilde r_{Q,k} &\sim\mathcal N\!\big(\sqrt{2\gamma_s}\sin\theta_k,\,1\big).
    \label{eq:rik_rqk_distributions}
\end{align}

Since $n_{I,k}$ and $n_{Q,k}$ are independent, this implies that $\tilde r_{I,k}$ and $\tilde r_{Q,k}$ will also be independent. Hence from \cref{Eq:lambda_k_appendix}, the MGF of $\lambda_k$, i.e., the expectation of $e^{t\lambda_k}$ factors into the product of two expectations, one for each component given by
\begin{align}
    M_{\lambda_k}(t \mid \theta_k)
    &= \mathbb E\!\left[e^{t\lambda_k}\Big \vert \theta_k\right]  \\
    &=\mathbb E\!\left[e^{t\tilde r_{I,k}^2}\Big \vert \theta_k\right]
     \mathbb E\!\left[e^{-t\tilde r_{Q,k}^2}\Big \vert \theta_k\right].
    \label{Eq:MGF_lambda_intermediate}
\end{align}
Note that each factor in \cref{Eq:MGF_lambda_intermediate} is the MGF of the square of a unit-variance Gaussian random variable. We therefore make use of the following standard result to evaluate them.

Let $X\sim\mathcal N(m,1)$. Then $X^2$ is a non-central chi-square random  variable with one degree of freedom and non-centrality parameter $m^2$, and its MGF is given by \cite[p.~522]{hogg2019introduction}
\begin{equation}
    \mathbb E\!\left[e^{sX^2}\right]=\frac{1}{\sqrt{1-2s}}\exp\!\left(\frac{m^2 s}{1-2s}\right),
    \quad s<\tfrac12.
    \label{Eq:MGF_square_gaussian}
\end{equation}
Note that the restriction $s<\tfrac12$ is necessary, since the integrand
$e^{sx^2}e^{-(x-m)^2/2}$ decays as $e^{-(1-2s)x^2/2}$ and is integrable
only when $1-2s>0$.

We apply \cref{Eq:MGF_square_gaussian} to $\tilde r_{I,k}$ with $m=\sqrt{2\gamma_s}\cos\theta_k$ and $s=t$, and to $\tilde r_{Q,k}$ with $m=\sqrt{2\gamma_s}\sin\theta_k$ and $s=-t$ to get
\begin{align}
    \mathbb E\!\left[e^{t\tilde r_{I,k}^2}\mid\theta_k\right]
    &=\frac{1}{\sqrt{1-2t}}\exp\!\left(\frac{2\gamma_s t\cos^2\theta_k}{1-2t}\right),
    \label{eq:MGF_rI}\\
    \mathbb E\!\left[e^{-t\tilde r_{Q,k}^2}\mid\theta_k\right]
    &=\frac{1}{\sqrt{1+2t}}\exp\!\left(-\frac{2\gamma_s t\sin^2\theta_k}{1+2t}\right).
    \label{eq:MGF_rQ}
\end{align}
Since \cref{eq:MGF_rI} requires the condition $t<1/2$ and \cref{eq:MGF_rQ} requires the condition $-t<1/2$, both require that $|t|<1/2$.
Substituting \cref{eq:MGF_rI} and \cref{eq:MGF_rQ} into \cref{Eq:MGF_lambda_intermediate} and using $(1-2t)(1+2t)=1-4t^2$, we obtain
\begin{align}
    M_{\lambda_k}(t\mid\theta_k)
    =\frac{1}{\sqrt{1-4t^2}}
    \exp\!\left\{2\gamma_s t\left[\frac{\cos^2\theta_k}{1-2t}-\frac{\sin^2\theta_k}{1+2t}\right]\right\},
\end{align}
and this completes the proof of the lemma.
\end{IEEEproof}

\textit{Proof of \cref{Thm:MGF_QPSK_BPSK}}:
For QPSK, $\theta_k\in\{0,\pi/2,\pi,3\pi/2\}$ (see \cref{Eq:All_phases_MPSK}). For $\theta_k\in\{0,\pi\}$ we have $\cos^2\theta_k=1$ and $\sin^2\theta_k=0$, while for $\theta_k\in\{\pi/2,3\pi/2\}$ we have $\cos^2\theta_k=0$ and $\sin^2\theta_k=1$. Since all $\theta_k$ are equally likely, averaging \cref{eqn:MGF_QPSK_BPSK} over these two classes gives the required expression. Similarly, for BPSK we have $\theta_k\in\{0,\pi\}$, for which $\cos^2\theta_k=1$ and $\sin^2\theta_k=0$, which gives the required expression of the theorem. 
\hfill $\blacksquare$

\section{Proof of \cref{Thm:PDF_QPSK_BPSK}}
\label{Appendix:B_PDF_CUSUM}

In this appendix, we prove \cref{Thm:PDF_QPSK_BPSK}. Toward this, in the following proposition we first derive the conditional PDF of $\lambda_k$ when the phase of the transmitted symbol $\theta_k$ is given. Using this proposition we then derive the PDF of $\lambda_k$ when the data is generated by QPSK or BPSK.

\begin{proposition}
\label{Proposition:Conditional_PDF}
Let $\theta_k$ be the phase of the transmitted symbol.
Then the conditional PDF of the approximate LLR $\lambda_k$ defined in \cref{eq:approx_llr} is given by
\begin{equation}
\begin{aligned}
    f_{\lambda_k}(x\mid \theta_k )
    ={}&
    e^{-\gamma_s}
    \sum_{n=0}^{\infty}
    \sum_{m=0}^{2n}
    \frac{\binom{2n}{m}}
    {\pi(2n)!\,2^{2n+1}}
    (2\gamma_s)^n
    \\
    &\times
    \tan^{m}\left(\frac{\pi}{4}-\theta_k\right)
    \left(1+\sin 2\theta_k\right)^n
    \\
    &\times
    x^{2n-m}|x|^{m-n}
    K_{m-n}\left(\frac{|x|}{2}\right),
\end{aligned}
\label{eq:pdf_general}
\end{equation}
where $K_{\nu}(\cdot)$ denotes the modified Bessel function of the second kind
of order $\nu$.
\hfill $\square$
\end{proposition}
%
%
%
%
\begin{IEEEproof} 
The key idea is to write $\lambda_k$ as the
product of two independent Gaussian random variables, use the closed-form PDF of
such a product, and simplify the resulting expression for our parameters.

We begin by simplifying $\lambda_k$. As shown in \cref{Eq:lambda_k_appendix} of Appendix~A, $\lambda_k$ is given by
\begin{equation}
\begin{aligned}
    \lambda_k
    =\tilde r_{I,k}^2-\tilde r_{Q,k}^2
    &=(\tilde r_{I,k}+\tilde r_{Q,k})(\tilde r_{I,k}-\tilde r_{Q,k})\\
    &=X_kY_k,    \label{Eq:lambda_k_factored}
\end{aligned}
\end{equation}
where in the last step we have defined $X_k=\tilde r_{I,k}+\tilde r_{Q,k}$ and $Y_k=\tilde r_{I,k}-\tilde r_{Q,k}$. 

We next find the joint distribution of $X_k$ and $Y_k$. As discussed in the proof of \cref{Proposition:MGF_QPSK_BPSK} in Appendix~A, 
for the given $\theta_k$, $\tilde r_{I,k}$ and $\tilde r_{Q,k}$ are independent
Gaussian random variables with unit variance and means
$\sqrt{2\gamma_s}\cos\theta_k$ and $\sqrt{2\gamma_s}\sin\theta_k$, respectively
(see \cref{eq:rik_rqk_distributions}).
Since $X_k$ and $Y_k$ are linear combinations of
$\tilde r_{I,k}$ and $\tilde r_{Q,k}$, they are jointly Gaussian given $\theta_k$.
Let $\mu_X,\mu_Y$ and $\sigma_X^2,\sigma_Y^2$ denote the means and variances of
$X_k$ and $Y_k$, respectively, conditioned on $\theta_k$. By linearity of
expectation we get,
\begin{equation}
\begin{aligned}
    \mu_X&=\sqrt{2\gamma_s}\,(\cos\theta_k+\sin\theta_k)\\ 
    \mu_Y&=\sqrt{2\gamma_s}\,(\cos\theta_k-\sin\theta_k).
    \label{Eq:mu_XY}
\end{aligned}
\end{equation}
Since $\tilde r_{I,k}$ and $\tilde r_{Q,k}$ are independent, the variance of their
sum (or difference) is the sum of their variances, so
$\sigma_X^2=\sigma_Y^2=1+1=2$. 
Since the covariance is
bilinear and symmetric we have,
\begin{align*}
    \operatorname{Cov}(X_k,Y_k)
    &=\operatorname{Cov}(\tilde r_{I,k}+\tilde r_{Q,k},\,\tilde r_{I,k}-\tilde r_{Q,k})\\
    &=\operatorname{Var}(\tilde r_{I,k})-\operatorname{Cov}(\tilde r_{I,k},\tilde r_{Q,k}) \\
    &\quad +\operatorname{Cov}(\tilde r_{Q,k},\tilde r_{I,k})-\operatorname{Var}(\tilde r_{Q,k})\\
    &=\operatorname{Var}(\tilde r_{I,k})-\operatorname{Var}(\tilde r_{Q,k})\\
    &=1-1=0,
\end{align*}
where the two cross terms cancel by symmetry of the covariance, and the
last step uses that both components have unit variance. For jointly Gaussian random variables, zero covariance implies independence.
Thus in \cref{Eq:lambda_k_factored}, $\lambda_k$ is product of two
independent Gaussian random variables $X_k$ and $Y_k$. 

To obtain the PDF of product of $X_k$ and $Y_k$, we make use of the following result.
\begin{lemma}
\label{Lemma:Gaussian_product_PDF}
Let $U\sim\mathcal N(\mu_U,\sigma_U^2)$ and $V\sim\mathcal N(\mu_V,\sigma_V^2)$
be independent random variables. Then the PDF of $Z=UV$ is given by
\cite[Eq.~(13)]{gaussian_product_pdf_2016}
\begin{equation}
\begin{aligned}
f_{Z}&(x)
=
\exp\!\left(
-\frac{\mu_U^2}{2\sigma_U^2}
-\frac{\mu_V^2}{2\sigma_V^2}
\right)
\sum_{n=0}^{\infty}\sum_{m=0}^{2n}
\binom{2n}{m}
\\
&\times
\frac{
\mu_U^m\mu_V^{2n-m}\,
x^{2n-m}|x|^{m-n}
}{
\pi(2n)!\,
\sigma_U^{n+m+1}
\sigma_V^{3n-m+1}
}
K_{m-n}\!\left(\frac{|x|}{\sigma_U\sigma_V}\right),
\label{Eq:Lemma_UV_PDF}
\end{aligned}
\end{equation}
where $K_\nu(\cdot)$ denotes the modified Bessel function of the second kind of
order $\nu$ \cite[Ch.~10]{olver2010nist}.
\hfill $\square$
\end{lemma}

We now apply \cref{Lemma:Gaussian_product_PDF} with $U=X_k$ and
$V=Y_k$, $\mu_U=\mu_X$, $\mu_V=\mu_Y$, and
$\sigma_U=\sigma_V=\sqrt{2}$.
It can be seen the the exponential factor of \cref{Eq:Lemma_UV_PDF} reduces to
\begin{align}
\exp\!\left(
-\frac{\mu_U^2}{2\sigma_U^2}
-\frac{\mu_V^2}{2\sigma_V^2}
\right)
=
\exp\!\left(-\frac{\mu_X^2+\mu_Y^2}{4}\right)
=e^{-\gamma_s},
\label{Eq:Lemma_exp}
\end{align}
since $\mu_X^2+\mu_Y^2=4\gamma_s$ by
\cref{Eq:mu_XY}. We next simplify the remaining 
terms in \cref{Eq:Lemma_UV_PDF} 
to get
\begin{align}
\sigma_U^{n+m+1}\sigma_V^{3n-m+1}&=2^{2n+1} \\  
\frac{|x|}{(\sigma_U\sigma_V)}
&= \frac{|x|}{2}
\end{align}
We next simplify the term $\mu_U^m\mu_V^{2n-m}$ in \cref{Eq:Lemma_UV_PDF}. We first
separate the powers of $\mu_V$ as
\begin{align}
\mu_U^m\mu_V^{2n-m}
=\mu_V^{2n}\left(\frac{\mu_U}{\mu_V}\right)^{m}
=\mu_X^{2n}\left(\frac{\mu_Y}{\mu_X}\right)^{m}.
\label{Eq:Lemma_muUV_split}
\end{align}
We now substitute \cref{Eq:mu_XY} into the two factors. For the first factor,
\begin{align}
\mu_X^{2n}=\left(2\gamma_s\right)^n(\cos\theta_k+\sin\theta_k)^{2n}
=(2\gamma_s)^n(1+\sin2\theta_k)^n,
\label{eq:lemma_muUV_intermediate_ux}
\end{align}
where we have used $(\cos\theta_k+\sin\theta_k)^2=1+\sin2\theta_k$. For the second factor,
\begin{align}
\frac{\mu_Y}{\mu_X}
=\frac{\cos\theta_k-\sin\theta_k}{\cos\theta_k+\sin\theta_k}
=\frac{1-\tan\theta_k}{1+\tan\theta_k}
=\tan\!\left(\frac{\pi}{4}-\theta_k\right),
\label{eq:lemma_muUV_intermediate_ux/uy}
\end{align}
where the second step divides the numerator and denominator by $\cos\theta_k$. Hence, by substituting \cref{eq:lemma_muUV_intermediate_ux} and \cref{eq:lemma_muUV_intermediate_ux/uy} in \cref{Eq:Lemma_muUV_split},
\begin{align}
\mu_U^m\mu_V^{2n-m}=(2\gamma_s)^n(1+\sin2\theta_k)^n
\tan^{\,m}\!\left(\frac{\pi}{4}-\theta_k\right).
\label{Eq:Lemma_muUV}
\end{align}
Note that the ratio $\mu_Y/\mu_X$ is well defined for the phases considered in this work.
Specifically, $\cos\theta_k+\sin\theta_k=0$ only when $\theta_k=3\pi/4+j\pi$ for an
integer $j$, which does not belong to the set of phases
$\{0,\pi/2,\pi,3\pi/2\}$ of QPSK or $\{0,\pi\}$ of BPSK. Hence, $\mu_X\neq0$ and
\cref{Eq:Lemma_muUV} holds for every $\theta_k$ of interest.
%
Substituting equations~(\ref{Eq:Lemma_exp}) and (\ref{Eq:Lemma_muUV}) in \cref{Lemma:Gaussian_product_PDF}, we obtain
\begin{equation}
\begin{aligned}
    f_{\lambda_k}(x\mid \theta_k )
    ={}&
    e^{-\gamma_s}
    \sum_{n=0}^{\infty}
    \sum_{m=0}^{2n}
    \frac{\binom{2n}{m}}
    {\pi(2n)!\,2^{2n+1}}
    (2\gamma_s)^n
    \\
    &\times
    \tan^{m}\left(\frac{\pi}{4}-\theta_k\right)
    \left(1+\sin 2\theta_k\right)^n
    \\
    &\times
    x^{2n-m}|x|^{m-n}
    K_{m-n}\left(\frac{|x|}{2}\right),
\end{aligned}
\end{equation}
which is the required expression of the lemma.
\end{IEEEproof}
\textit{Proof of \cref{Thm:PDF_QPSK_BPSK}}:
We specialize \cref{Proposition:Conditional_PDF} to each case. For QPSK ($M^{\prime}=4$), the phase $\theta_k$ takes one of the four equally likely values $\{0,\pi/2,\pi,3\pi/2\}$ (see \cref{Eq:All_phases_MPSK}). Note that in \cref{eq:pdf_general}, $\theta_k$ enters only through $\tan^{m}\left(\frac{\pi}{4}-\theta_k\right)\left(1+\sin 2\theta_k\right)^n$. For all four phases we have $\sin2\theta_k=0$, and hence $\left(1+\sin 2\theta_k\right)^n=1$. Further, since the tangent has period $\pi$, we have
\begin{equation}
\tan\!\left(\frac{\pi}{4}-\theta_k\right)=
\begin{cases}
+1, & \theta_k\in\{0,\pi\},\\
-1, & \theta_k\in\{\pi/2,3\pi/2\}.
\end{cases}
\end{equation}
Thus, the conditional PDFs for $\theta_k\in\{0,\pi\}$ are identical, and similarly for $\theta_k\in\{\pi/2,3\pi/2\}$. Averaging over the four equally likely phases, the only $\theta_k$-dependent term in \cref{eq:pdf_general} becomes
\begin{equation}
\frac{1}{4}\sum_{\theta_k}\tan^{m}\!\left(\frac{\pi}{4}-\theta_k\right)
=\frac{1}{4}\left(2+2(-1)^m\right)=\frac{1+(-1)^m}{2},
\end{equation}
which is $1$ for even $m$ and $0$ for odd $m$. Hence, only the terms with even $m$ survive. For even $m$, $x^{2n-m}$ is an even power of $x$, and therefore $x^{2n-m}|x|^{m-n}=|x|^{2n-m}|x|^{m-n}=|x|^{n}$. Substituting $m=2j$ with $j=0,\dots,n$ and using $K_{2j-n}=K_{n-2j}$, we obtain \cref{eq:thm_qpsk_pdf}.

For BPSK, the phase takes one of the two equally likely values $\{0,\pi\}$. For both phases we have $\sin2\theta_k=0$ and $\tan(\pi/4-\theta_k)=1$. Hence, both phases give the same conditional PDF, which therefore also equals the PDF of $\lambda_k$ under BPSK. Substituting these values in \cref{eq:pdf_general} gives \cref{eq:thm_bpsk_pdf}. \hfill $\blacksquare$

\section{Proof of \cref{thm:Gaussian_Approx}}
\label{Appendix:mosum_approximation_proof}
From \cref{Eq:MOSUM_stat_lambda_vec}) recall that $\mathbf{Y}_n$ is defined as 
\begin{equation}
    \mathbf{Y}_n
    =
    \begin{bmatrix}
        Y_{w}^\lambda &
        Y_{w+1}^\lambda &
        \cdots &
        Y_{w+n-1}^\lambda
    \end{bmatrix}^{T}.
    \label{Eq:MOSUM_stat_lambda_vec_appendix}
\end{equation}
where each component $Y^\lambda_k$ in $\mathbf{Y}_n$ for $k=w, w+1, \ldots, w+n-1$ is given by
\begin{equation}
    Y^\lambda_k
    =
    \sum_{k^{\prime}=k-w+1}^{k}\lambda_{k^{\prime}}.
    \label{Eq:MOSUM_stat_lambda_appendix}
\end{equation}
Since the mean of $\lambda_{k^{\prime}}$ is $\mu$, from \cref{Eq:MOSUM_stat_lambda_appendix} we have $\mathbb{E}[Y_k^{\lambda}] = w\mu$ and hence mean vector of $\mathbf{Y}_n$ is $w\mu \mathbf{1}_n$.

We now derive the covariance matrix of $\mathbf{Y}_n$. Firstly note that, since the $\lambda_{k^{\prime}}$ are independent,
\begin{equation}
    \operatorname{Cov}(\lambda_{\scriptstyle k_1},\lambda_{\scriptstyle k_2})
    =
    \begin{cases}
        \sigma^2, & k_1=k_2,\\
        0, & k_1\neq k_2.
    \end{cases}
\end{equation}
Now consider two MOSUM statistics $Y_{i}^{\lambda}$ and $Y_{j}^{\lambda}$. Since $\lambda_{k^{\prime}}$ are independent, 
covariance of $Y_{i}^{\lambda}$ and $Y_{j}^{\lambda}$ is completely determined by the number of $\lambda_{k^{\prime}}$ that are common to $Y_{i}^{\lambda}$ and $Y_{j}^{\lambda}$. 
Suppose $d=|i-j|$. When $d<w$, 
$Y_{i}^{\lambda}$ and $Y_{j}^{\lambda}$ will have $w-d$ common terms and when $d\geq w$ they will have no common terms 
and this implies that
\begin{equation}
\begin{aligned}
    [\boldsymbol{\Sigma}_k]_{i,j} = \operatorname{Cov}
    \left(
        Y_i^{\lambda},
        Y_j^{\lambda}
    \right)
    =
\begin{cases}
        (w-d)\sigma^2 &\mbox{~if~} d<w,\\
        0,  &\mbox{~if~} d\geq w,
    \end{cases}
\end{aligned}
\end{equation}
and this completes the proof. 
\hfill $\blacksquare$

\bibliographystyle{IEEEtran}
\bibliography{References}

\end{document}